\documentclass[journal, twoside, web]{ieeecolor}
\usepackage{generic}
\usepackage{tikz}
\usepackage{amsmath}
\usepackage{graphicx}
\usepackage{amssymb}
\usepackage{mathtools}
\usepackage{booktabs,amsfonts,dcolumn}
\usepackage{textcomp}
\usepackage{hyperref}
\usepackage{accents}
\newcommand{\0}{\mathbf{0}}
\newcommand{\I}{\mathbf{I}}
\newcommand{\ep}{\mathtt{e}}
\newcommand{\df}{\mathrm{d}}
\renewcommand{\min}{\mathtt{min}}

\newcommand{\diag}{\mathtt{diag}}
\renewcommand{\top}{\mathsf{T}}
\newcommand{\R}{\mathbb{R}}

\DeclareRobustCommand{\ubar}[1]{\underaccent{\bar}{#1}}
\usepackage{upgreek}
\usepackage{cmupint}
\usepackage{mathdots}

\usepackage{amsthm}
\theoremstyle{definition}

\newtheorem{problem}{\bf Problem}

\newtheorem{lemma}{\bf Lemma}

\newtheorem{definition}{\bf Definition}
\newtheorem{proposition}{\bf Proposition}

\DeclareRobustCommand{\uppartial}{\text{\rotatebox[origin=t]{20}{\scalebox{0.95}[1]{$\partial$}}}\hspace{-1pt}}

\DeclarePairedDelimiter{\norm}{\lVert}{\rVert}
\DeclarePairedDelimiter{\abs}{\lvert}{\rvert}

\usepackage{flushend}
\usepackage{balance}

\title{Safety Control of Uncertain MIMO Systems Using Dynamic Output Feedback Barrier Pairs}
\author{Binghan He and Takashi Tanaka
\thanks{This work was supported by the Air Force Office of Scientific Research under Grant FA9550-20-1-0101. \emph{(Corresponding author: Binghan He.)}}
\thanks{Binghan He was with the Oden Institute for Computational Engineering and Sciences, The University of Texas at Austin, Austin, TX 78712 USA. He is now with the Department of Electrical Engineering and Computer Sciences, University of California, Berkeley, CA 94720 USA (e-mail: binghan.he@berkeley.edu).}
\thanks{Takashi Tanaka is with the Department of Aerospace Engineering and Engineering Mechanics, The University of Texas at Austin, Austin, TX 78712 USA (e-mail: ttanaka@utexas.edu).}}

\newcommand\copyrighttext{\footnotesize\sf This article has been accepted for publication in IEEE Transactions on Automatic Control. \copyright \ 2024 IEEE. DOI: 10.1109/TAC.2024.3462288}

\newcommand\copyrightnotice{\begin{tikzpicture}[remember picture,overlay] \node[anchor=north, xshift=0pt] at (current page.north) {\fbox{\parbox{\dimexpr\textwidth-\fboxsep-\fboxrule\relax}{\copyrighttext}}}; \end{tikzpicture} \vspace{-10pt}}

\begin{document}

\maketitle

\copyrightnotice

\begin{abstract}
Safety control of dynamical systems using barrier functions relies on knowing the full state information. 
This paper introduces a novel approach for safety control in uncertain MIMO systems with partial state information. 
The proposed method combines the synthesis of a vector norm barrier function and a dynamic output feedback safety controller to ensure robust safety enforcement. 
The safety controller guarantees the invariance of the barrier function under uncertain dynamics and disturbances. 
To address the challenges associated with safety verification using partial state information, a barrier function estimator is developed. 
This estimator employs an identifier-based state estimator to obtain a state estimate that is affine in the uncertain model parameters of the system. 
By incorporating a priori knowledge of the limits of the uncertain model parameters and disturbances, the state estimate provides a robust upper bound for the barrier function. 
Comparative analysis with existing control barrier function based methods shows the advantage of the proposed approach in enforcing safety constraints under tight input constraints and the utilization of estimated state information.
\end{abstract}

\section{Introduction}

For dynamical systems that involve disturbances and uncertain dynamics, such as a wearable robot coupled with time-varying human dynamics \cite{thomas2021formulating} or a self-driving vehicle deployed in an uncertain environment \cite{sadigh2018planning}, safety is a critical issue. 
While conventional robust and adaptive control techniques have demonstrated effectiveness in controlling these uncertain dynamical systems \cite{ioannou1996robust}, recent advancements in reinforcement learning, iterative learning, and leaning-based model predictive control hold promise for enhanced adaptation and performance \cite{kabzan2019learning, ahn2020data, cheng2023practice}. 
However, the application of these methods introduces potential risks and safety concerns \cite{brunke2022safe}. 
Morevover, in systems involving human-in-the-loop, it is crucial to consider the control authority and performance requirements of the human operator \cite{alonso2018system}. 
To strike a balance between safety and performance, a dedicated safety control framework is required to selectively correct unsafe control commands, ensuring safety while accommodating learning-based control methods and respecting human intentions.

Typically, barrier functions are employed to ensure the satisfaction of safety constraints \cite{prajna2004safety} or establish a safety controller \cite{wieland2007constructive} for a dynamical system. 
However, determining barrier function values often requires full state information, which is not always available for uncertain dynamical systems. 
To address these challenges, we propose a novel safety control framework specifically designed for multiple-input multiple-output (MIMO) systems. 
Our framework tackles the safety control problem for an uncertain MIMO system with only partial state information, overcoming the limitations of existing approaches.

\subsection{Related Works}

While there are multiple methods to synthesize a barrier function and a safety controller that enforces the barrier function, many of them \cite{tee2009barrier, liu2016barrier} consider the barrier function to be a Lyapunov function.
In particular, the Lyapunov function needs to provide an invariant set that is also a subset of the safe region in state space.
If the full state information of the system is available, we can formulate the barrier function as a quadratic Lyapunov function and the safety controller as a full state feedback controller \cite{hu2003composite, thomas2018safety}. 
Then, the synthesis of the barrier function and safety controller can be expressed as a linear matrix inequality (LMI) optimization problem \cite{boyd1994linear}, which turns the safety constraints and the Lyapunov stability criterion into a set of LMI constraints.
Although this type of synthesis method requires a linear time-invariant system model, we can also apply it to a nonlinear or time-varying system by formulating the system as a linear model in feedback with a polytopic or norm-bound uncertainty \cite{thomas2019quadric}.

For uncertain systems with partial state information, a controller parameter transformation method has been proposed in \cite{scherer1997multiobjective}, which converts bilinear matrix inequality (BMI) constraints for synthesizing a dynamic output feedback controller into LMI constraints.
This controller parameter transformation method can be adopted for solving $\mathrm{H_2}$, $\mathrm{H_{\infty}}$ \cite{doyle1989state}, and bounded-input bounded-output (BIBO) \cite{dahleh1988necessary} control problems.
However, the dynamic output feedback controller synthesized using the method in \cite{scherer1997multiobjective} does not guarantee the inclusion of a state observer, which plays an essential role in detecting the potential safety issues in a system with partial state information.
In fact, the synthesis of an observer-based controller is a more difficult problem. 
This type of dynamic output feedback controller is usually obtained by directly solving a BMI problem \cite{lens2008observer} or converting the BMI problem to an LMI problem through a relaxation that causes additional conservatism \cite{lien2004robust}.

Using a state observer, a residual signal or a limit monitoring function in terms of the residual signal can be generated to detect a fault signal that causes the malfunction of an uncertain system \cite{ding2003threshold}.
Based on the controller parameter transformation method in \cite{scherer1997multiobjective}, an LMI approach is proposed in \cite{li2012dynamic} for finding the limit monitoring function using a dynamic observer \cite{park2002dynamic}, which is a state observer in feedback with an additional dynamic sub-system. 
However, unlike the barrier functions, the residual signals and the limit monitoring functions in \cite{ding2003threshold, li2012dynamic} do not detect the potential safety issues in relation to any state space constraints.

The combination of the control barrier function (CBF) framework with quadratic programming (QP) has emerged as a popular approach for enforcing safety in constrained control systems. 
Numerous studies, including \cite{xu2015robustness, nguyen2016exponential, ames2017control, xu2018constrained}, have demonstrated the effectiveness of this method in ensuring safety. 
However, despite its advantages, the CBF-based QP approach encounters specific challenges when dealing with tight input constraints or when utilizing estimated state information. 
Particularly, under conditions of input saturation, the method may struggle to enforce the desired constraints, even when provided with accurate full state information. 
Additionally, the inclusion of estimated state information introduces further complexities, as evidenced in \cite{wang2022observer}, where the conventional CBF-based QP method fails to satisfy safety limits in the absence of input saturation. 
These limitations underscore the need for advancements in safety control methodologies capable of effectively managing input constraints and ensuring robust enforcement of safety constraints, even in the presence of uncertainties and disturbances.

In \cite{he2020robust}, the barrier function of an uncertain dynamical system is defined as a vector norm function. 
Through the triangle inequality of a vector norm function, we can find a barrier function's upper bound, which is in terms of a state estimate obtained from a state observer.
In \cite{casavola2005robust}, a similar approach is used to estimate the upper bound for a limit monitoring function that detects a fault signal.
By finding the upper bound for the vector norm barrier function, we can figure out when to trigger a safety controller to avoid the potential safety issues of the original system.
The main reason why the upper bound for the vector norm barrier function can be found in \cite{he2020robust} is because the state observer is in the form of an identifier-based estimator \cite{morse1980global}. 
This special type of state observer was originally developed for the purpose of system identification and adaptive control \cite{datta1996adaptive, morse1996supervisory}.
As a byproduct, it also provides a robust state estimate under the presence of model uncertainty.

To form the identifier-based estimator in \cite{morse1980global}, we need to convert the state space system model into an observable canonical form. 
While finding an observable canonical form for a single-input single-output (SISO) system is straightforward, there is no guarantee that a minimal realization in the observable canonical form also exists for a MIMO system. 
In \cite{morse1994mimo}, a MIMO identifier-based estimator is proposed based on an output injection canonical form, which provides minimal realizations for MIMO systems \cite{luenberger1967canonical}. 
While converting a state space system model into an output injection canonical form does not guarantee the preservation of the physical meaning from the original model, system identification methods in \cite{pan1996parameter, romano2016matchable} offer paradigms for modeling an uncertain MIMO system in output injection canonical form.
However, the state estimate generated using the identifier-based estimator cannot be directly used to form an observer-based dynamic output feedback controller.
Therefore, the safety control synthesis method in \cite{he2020robust} requires the stability of the uncontrolled system or the existence of a stable static output feedback controller \cite{crusius1999sufficient}. 

\subsection{Contributions}

Inspired by these previous works, we propose our new safety control framework, which aims to solve the safety control problem for an uncertain MIMO system with partial state information. 
The main contributions of this paper are summarized as follows.
\begin{itemize}
\item[(1)] 
To enforce the safety constraints of the uncertain MIMO system, we develop a synthesis method that simultaneously creates a vector norm barrier function and a dynamic output feedback safety controller. 
Using the controller parameter transformation scheme in \cite{scherer1997multiobjective}, this safety controller guarantees the invariance of our barrier function with bounded model uncertainty and disturbance.
While both our synthesis method and the approach in \cite{li2012dynamic} employ a similar parameter transformation scheme, they diverge in their objectives. 
The method proposed in \cite{li2012dynamic} primarily emphasizes state estimation and fault detection, neglecting the enforcement of state space constraints. 
In contrast, our approach prioritizes the enforcement of state space constraints, while employing a dedicated state estimator within our safety control framework to handle the state estimation task separately.
\item[(2)] 
To detect the potential safety issues of the uncertain MIMO system, we develop a robust barrier function estimator that incorporates a MIMO identifier-based estimator inspired by \cite{morse1994mimo}.
Similar to the method in \cite{he2020robust}, this identifier-based estimator provides us with a robust state estimate, which helps us find the upper bound for our vector norm barrier function.
Expanding upon the work of \cite{morse1994mimo}, we introduce an LMI approach to synthesize the MIMO identifier-based estimator, aiming to achieve a tight upper bound for our vector norm barrier function.
\item[(3)] 
Through comprehensive simulation examples, we illustrate the remarkable capabilities of our approach in detecting potential safety violations and appropriately engaging the safety controller to mitigate risks. 
The key advantage lies in the robust triggering mechanism of our barrier function estimator, which enables the safety controller to intervene effectively when necessary, ensuring robust compliance with safety constraints. 
Notably, this stands in contrast to CBF-based methods, such as \cite{xu2015robustness, nguyen2016exponential, ames2017control, xu2018constrained}, which may face challenges in handling tight input constraints or utilizing estimated state information, leading to potential safety lapses.
\end{itemize}
In the preliminary version of this work \cite{he2023barrier}, we introduced our safety control framework for uncertain SISO systems with partial state information. 
This paper builds significantly upon our results in \cite{he2023barrier}. 
First, this paper extends our safety control framework from uncertain SISO systems in \cite{he2023barrier} to uncertain MIMO systems. 
In particular, this paper shows how to find an upper bound for our barrier function using a MIMO identifier-based estimator for an uncertain MIMO system.
Second, this paper proposes a composite barrier function consisting of multiple vector norm barrier functions for the closed-loop system with our safety controller. 
Compared to the single vector norm barrier function in \cite{he2023barrier}, the composite barrier function in this paper provides a larger invariant set inside the safe state-space region.

\section{Preliminaries}

In this section, we introduce models of an uncertain MIMO system $\Sigma_p$ and a full-order dynamic output feedback controller $\Sigma_k$. 
The uncertain MIMO system is modeled as a diagonal norm-bound linear differential inclusion (DNLDI) \cite{boyd1994linear}. 
Then we give an overview of barrier pairs, which will be used later for safety verification and control of the uncertain MIMO system. 
Based on the concept of barrier pairs, we present our formal problem statement.

\subsection{Notations} \label{sec:notations}

In this paper, we define $S_{\star}$ as a selection matrix, where $\star$ represents a specific subscript. 
For an $n$-th order MIMO system $\Sigma_p$, we define two specific selection matrices as 
\begin{equation*}
S_p \triangleq \begin{bmatrix} \I_{n \times n} & \0_{n \times n} \end{bmatrix}, \quad
S_k \triangleq \begin{bmatrix} \0_{n \times n} & \I_{n \times n} \end{bmatrix},
\end{equation*}
which extract the plant state $x_p \in \R ^ {n}$ and the controller state $x_k \in \R ^ {n}$ from the closed-loop state vector $x_\mathtt{CL} \in \R ^ {2n}$.

In \cite{morse1994mimo, pan1996parameter}, a canonical form of state space model for the MIMO system $\Sigma_p$ is derived according to a set of observability indices $\{ n_1, \, n_2, \, \dots, \, n_{n_y} \}$ for $\Sigma_p$, where $n = \sum_{i=1}^{n_y} n_i$. 
Based on this set of $n_{y}$ observability indices, $\Sigma_p$ is partitioned as $n_y$ coupled sub-systems $\Sigma_{p}^{(1)}, \, \Sigma_{p}^{(2)}, \, \dots, \, \Sigma_{p}^{(n_y)}$. 
In accordance with the notations introduced in \cite{pan1996parameter}, we adopt a superscript $(i)$ to delineate variables and model parameters associated with the $i$-th partition $\Sigma_p^{(i)}$ of $\Sigma_p$, for $i = 1, \, 2, \, \dots, \, n_y$. 

For a vector $x \in \mathbb{R}^{n_{x}}$, we define $x_{(j)} \in \R$ as its $j$-th element. 
For a matrix $A_x \in \mathbb{R}^{n_{A} \times n_{x}}$, which is post-multiplied by $x$, we define $A_{x(j)} \in \R ^ {n_{A} \times 1}$ as its $j$-th column, such that $A_{x} x = \sum_{j=1}^{n_{x}} A_{x(j)} x_{(j)}$.

Supposing $P \succ \0$, we define
\begin{equation*} \notag
\norm{\star}_{P} \triangleq \sqrt{\star ^ \top P \star}
\end{equation*}
as a vector norm function based on $P$. 
In our LMIs, we define $\mathtt{He} \{ \star \} \triangleq \star + \star ^ \top$. 
In a matrix $\begin{bmatrix} A & \star \\ B & C \end{bmatrix}$, the notation $\star$ in the upper right corner represents the ellipsis indicating terms induced by symmetry.

\subsection{State Model} \label{sec:model}

First of all, let us consider an $n$-th order controllable and observable system $\Sigma_p$, featuring $n_u$ inputs, $n_y$ outputs, and a set of observability indices $\{ n_1, \, n_2, \, \dots, \, n_{n_y} \}$, where $n \geq n_y$ and $n = \sum_{i=1}^{n_y} n_i$. 
We define 
\begin{equation}
\begin{aligned}
& \bar{A} \triangleq \diag \big\{ \bar{A}^{(1)}, \, \bar{A}^{(2)}, \, \dots, \, \bar{A}^{(n_y)} \big\} \\ 
& \bar{C} \triangleq \diag \big\{ \bar{C}^{(1)}, \, \bar{C}^{(2)}, \, \dots, \, \bar{C}^{(n_y)} \big\}
\end{aligned}
\end{equation}
where $\bar{A}^{(i)} \in \R ^ {n_i \times n_i}$ and $\bar{C}^{(i)} \in \R ^ {1 \times n_i}$ are
\begin{equation}
\begin{aligned}
& 
\bar{A}^{(i)}
\triangleq 
\begin{bmatrix*}[l] 
0 & \cdots & 0 & 0 \\
1 & \cdots & 0 & 0 \\
\vdots & \ddots & \vdots & \vdots \\
0 & \cdots & 1 & 0
\end{bmatrix*} \\
&
\bar{C}^{(i)} 
\triangleq 
\begin{bmatrix} 0 & \cdots & 0 & 1 \end{bmatrix}.
\end{aligned}
\end{equation}
for $i = 1, \, 2, \, \cdots, \, n_y$.
Then, the state-space model of the system $\Sigma_p$ can be expressed as
\begin{equation} \label{eq:model-0} 
\begin{aligned}
\dot{x}_{p} & = (\bar{A} + \Theta_{y} \Theta_{c} \bar{C}) \, x_{p} + \Theta_{u} \, u + \, w, \\ 
y  & = \Theta_{c} \bar{C} \, x_{p} + v,
\end{aligned}
\end{equation}
where $x_{p} \in \R ^ {n}$, $u \in \R ^ {n_u}$, and $y \in \R ^ {n_y}$ are the vectors of states, inputs, and outputs, $w \in \R ^ {n}$, and $v \in \R ^ {n_y}$ are the vectors of disturbance and measurement noise, $\Theta_y \in \R ^ {n \times n_y}$, $\Theta_u \in \R ^ {n \times n_u}$, and $\Theta_c \in \R ^ {n_y \times n_y}$ are uncertain parameter matrices.
If we define $h \triangleq \begin{bmatrix} u \\ y \end{bmatrix} \in \R ^ {n_h}$ and $d \triangleq \begin{bmatrix} w \\ v \end{bmatrix} \in \R ^ {n_d}$ for $n_h \triangleq n_u + n_y$ and $n_d \triangleq n + n_y$, \eqref{eq:model-0} becomes
\begin{equation} \label{eq:model}
\begin{aligned}
\dot{x}_{p} & = \bar{A} x_{p}  
+ 
\overbracket{
\setlength\arraycolsep{2pt}
\begin{bmatrix} \Theta_{u} & \Theta_{y} \end{bmatrix}
}^{\Theta_{h}}
h
+
\setlength\arraycolsep{2pt}
\begin{bmatrix} \I & - \Theta_{y} \end{bmatrix}
d,
\\
y  & = \Theta_{c} \bar{C} \, x_{p} + v,
\end{aligned}
\end{equation}
where the output $y$ is injected into a combined input $h$.  
In \cite{morse1994mimo}, the state space model described by \eqref{eq:model-0} is referred to as the output injection canonical form. 
As detailed in \cite[Theorem 1]{morse1994mimo}, a key procedure for expressing $\Sigma_p$ in the structure of \eqref{eq:model-0} is to define a non-singular matrix of $\Theta_{c}$, ensuring the existence of its inverse denoted as $\Theta_{a} \triangleq \Theta_{c}^{-1}$. 
As a sufficient but not necessary condition, such a non-singular matrix of $\Theta_{c}$ is guaranteed to exist for all controllable and observable systems $\Sigma_p$ with $n \geq n_y$ if $\Theta_{c}$ is defined as a unit lower triangular matrix.

The method outlined in \cite{pan1996parameter} identifies the models of uncertain MIMO systems in output injection canonical form, estimating a set of $n_{\uptheta}$ scalar parameters that define $\Theta_{h}$ and $\Theta_{a}$. 
Using this approach, we determine the parameter uncertainties for $\Theta_{h}$ and $\Theta_{a}$ as
\begin{equation} \label{eq:ths-qip-1}
\begin{aligned}
\Theta_{h} & = A_{h} + B_{h} \Delta_{h} C_{h}, \quad\quad \Theta_a = A_{a} + B_{a} \Delta_{a} C_{a}, \\
\Delta_{h} & \triangleq \diag \big\{ \updelta_{h(1)}, \, \updelta_{h(2)}, \, \dots, \, \updelta_{h(n_{\uptheta}^{h})} \big\} \in \R ^ {n_{\uptheta}^{h} \times n_{\uptheta}^{h}} \\
\Delta_{a} & \triangleq \diag \big\{ \updelta_{a(1)}, \, \updelta_{a(2)}, \, \dots, \, \updelta_{a(n_{\uptheta}^{a})} \big\} \in \R ^ {n_{\uptheta}^{a} \times n_{\uptheta}^{a}}
\end{aligned}
\end{equation}
where $n_{\uptheta}^{a} + n_{\uptheta}^{c} = n_{\uptheta}$, $\abs{\updelta_{i(j)}} \leq 1$ for all $i = h, \, a$ and $j = 1, \, 2, \, \dots, \, {n_{\uptheta}^{i}}$. 
In \eqref{eq:ths-qip-1}, $A_{i}$ denotes the nominal value of $\Theta_i$, and $B_{i}$ and $C_{i}$ are matrices quantifying the magnitude of bounded uncertainty for $i = h, \, a$. 
In addition, by partitioning $A_{h}$ and $C_{h}$ as $A_{h} = \begin{bmatrix} A_{u} & A_{y} \end{bmatrix}$ and $C_{h} = \begin{bmatrix} C_{u} & C_{y} \end{bmatrix}$, we derive
\begin{equation} \label{eq:ths-qip-2}
\begin{aligned}
& \Theta_u = A_{u} + B_{h} \Delta_{h} C_{u}, \quad && \Theta_y = A_{y} + B_{h} \Delta_{h} C_{y}.
\end{aligned}
\end{equation}
While $\Theta_{c} = \Theta_{a} ^ {-1}$ is not directly identified through \cite{pan1996parameter}, an LMI method in \cite{el2002inversion} defines  $A_{c} \triangleq A_{a} ^ {-1}$ and minimizes the upper bounds of $B_{c}$ and $C_{c}$ in
\begin{equation} \label{eq:ths-qip-3}
\begin{aligned}
\Theta_{c} & = A_{c} + B_{c} \Delta_{a} C_{c}, 
\end{aligned}
\end{equation}
based on the identified values of $A_{a}$, $B_{a}$, and $C_{a}$ in \eqref{eq:ths-qip-1}.
Substituting \eqref{eq:ths-qip-2} and \eqref{eq:ths-qip-3} into \eqref{eq:model}, we obtain a DNLDI model of $\Sigma_p$.

Next, let the state-space model of a full-order dynamic output feedback controller $\Sigma_k$ be
\begin{alignat}{1}
\dot{x}_{k} & = A_k x_{k} + B_k y \label{eq:controller-siso} \\
          u & = C_k x_{k} \notag 
\end{alignat}
where $x_k \in \R ^ {n}$ is the controller state vector and $A_k \in \R ^ {n \times n}$, $B_k \in \R ^ {n \times n_y}$, and $C_k \in \R ^ {n_u \times n}$ are the controller parameters to be determined.

Finally, we combine $\Sigma_p$ and $\Sigma_k$ to get our closed-loop system $\Sigma_\mathtt{CL}$
\begin{alignat}{3}
\dot{x}_\mathtt{CL} 
& 
= 
\overbracket{
\begin{bmatrix*}
A_{p}               & A_{u} C_{k} \\
B_{k} A_{c} \bar{C} & A_{k}
\end{bmatrix*}
}^{A_\mathtt{CL}}
&& 
x_\mathtt{CL}
+ 
\begingroup
\setlength\arraycolsep{2.0pt}
\overbracket{
\begin{bmatrix}
\I &  \0  & B_{h} & A_{y} B_{c} \\ 
\0 & B_{k}  & \0    & B_{k} B_{c} 
\end{bmatrix}
}^{B_\mathtt{CL}}
\endgroup
&&
\begin{bmatrix}
d \\ p 
\end{bmatrix}
\notag \\
q                   
& 
= 
\begingroup
\underbracket{
\begin{bmatrix*}
C_{y} A_{c} \bar{C} & C_{u} C_{k} \\
C_{c} \bar{C} & \0
\end{bmatrix*}
}_{C_\mathtt{CL}}
\endgroup
&& 
x_\mathtt{CL} 
+ 
\begingroup
\setlength\arraycolsep{4.4pt}
\underbracket{
\begin{bmatrix}
\0 &  \0  & \0    & C_{y} B_{c}  \\
\0 & \0     & \0    & \0 
\end{bmatrix}
}_{D_\mathtt{CL}} 
\endgroup
&&
\begin{bmatrix}
d \\ p 
\end{bmatrix}
\notag  \\
p                   
& 
= 
\begin{bmatrix} 
\Delta_{h} & \0 \\ 
\0 & \Delta_{a} 
\end{bmatrix}
q 
\label{eq:model-CL}
\end{alignat}
where $A_{p} \triangleq \bar{A} + A_{y} A_{c} \bar{C}$, $x_\mathtt{CL} \triangleq \begin{bmatrix} x_{p} \\ x_{k} \end{bmatrix} \in \R ^ {2 n}$ is the close-loop state vector, $p \triangleq \begin{bmatrix} p_{h} \\ p_{c} \end{bmatrix} \in \R ^ {n_{\uptheta}}$ and $q \triangleq \begin{bmatrix} q_{h} \\ q_{c} \end{bmatrix} \in \R ^ {n_{\uptheta}}$ are the vectors of disturbance and output signals related to the parameter uncertainty of $\Theta_{h}$ and $\Theta_{c}$.

\subsection{Barrier Pairs}

The concept of barrier pairs \cite{thomas2018safety} originally describes the relationship of a barrier function and a full state feedback controller in a safety control problem. 
In this paper, we extend the definition of barrier pairs to output feedback systems.
\begin{definition} \label{def:bp}
For a dynamical system $\Sigma_p$ with input $u$ and output $y$, a barrier pair is a pair $(\mathbf{B},\ \Sigma_k)$ consisting of a barrier function $\mathbf{B}: \R ^ {n + n_k} \rightarrow \R$ and a feedback controller $\Sigma_k$ that satisfy the following conditions:
\begin{itemize}
\item[(a)] $\upvarepsilon \leq \mathbf{B} (x_\mathtt{CL}) \leq 1, \, u = \Sigma_k (y) \implies \dot{\mathbf{B}} (x_\mathtt{CL}) < 0$,
\vspace{3pt}
\item[(b)] $\mathbf{B} (x_\mathtt{CL}) \leq 1 \implies x_p \in \mathcal{X}_s, \ \Sigma_k (y) \in \mathcal{U}$,
\vspace{1pt}
\end{itemize}
where $x_p \in \mathcal{X}_s$ and $u \in \mathcal{U}$ represent the state and input constraints for the system $\Sigma_p$, $n$ and $n_k$ denote the orders of $\Sigma_p$ and $\Sigma_k$ respectively, $x_\mathtt{CL} \in \R ^ {n + n_k}$ denotes the close-loop state vector, and $0 \leq \upvarepsilon < 1$ determines the size of a residue set $\{ x_\mathtt{CL}: \ \mathbf{B} (x_\mathtt{CL}) \leq \upvarepsilon \}$.
\end{definition}
Definition~\ref{def:bp} generalizes the definition of the barrier pairs provided in \cite{thomas2018safety}. 
However, if we assign $n_k = 0$ and consider $y = x_p$ where $x_p$ represents the full state of $\Sigma_p$, Definition~\ref{def:bp} is transformed into an equivalent form of the barrier pair definition presented in \cite{thomas2018safety}.
In this paper, we specifically consider the case where $n_k = n$ as described in \eqref{eq:controller-siso}.

\subsection{Problem Statement}

\begin{figure}
\centering
\includegraphics[width=0.50\textwidth]{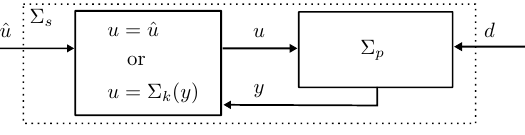}
\caption{In this paper, our switching system $\Sigma_s$ chooses either an original input $u = \hat{u}$ of $\Sigma_p$ or a known-to-be-safe input $u = \Sigma_k (y)$ based on a barrier function $\mathbf{B}$.}
\label{fig:problem}
\end{figure}

In this paper, we consider an uncertain dynamical system $\Sigma_p$ described in \eqref{eq:model-0} with the assumption that $\Sigma_p$ is observable and controllable for all $x_p \in \mathcal{X}_s$, the assumption $d \in \mathcal{D}$ and the safety constraints $x_p \in \mathcal{X}_s$ and $u \in \mathcal{U}$. 
The sets $\mathcal{X}_s$, $\mathcal{U}$, and $\mathcal{D}$ are defined as
\begin{alignat}{7}
& \mathcal{X}_s && \triangleq \{x_p \ & : \ && \abs{f_{i} ^ \top x_p} & \leq 1, && \ i = 1, \, 2, \, \cdots, \, n_s && \}, \notag \\
& \mathcal{U} && \triangleq \{ u \ & : \ && \abs{u_{(i)}} & \leq \bar{u}_{(i)}, && \ i = 1, \, 2, \, \cdots, \, n_u && \}, \label{eq:bars} \\
& \mathcal{D} && \triangleq \{ d \ & : \ && \abs{d_{(i)}} & \leq \bar{d}_{(i)}, && \ i = 1, \, 2, \, \cdots, \, n_d && \}. \notag
\end{alignat}
where $\mathcal{X}_s$ enforces $n_s$ state-space safety constraints on the system $\Sigma_p$, the vector $f_{i} \in \R ^ {n \times 1}$ denotes a hyperplane boundary condition that defines each constraint, and the sets $\mathcal{U}$ and $\mathcal{D}$ restrict the absolute values of the input $u_{(i)}$ and disturbance $d_{(i)}$ within specified bounds $\bar{u}_{(i)}$ and $\bar{d}_{(i)}$ for $i = 1, \, 2, \, \cdots, \, n_u$ and $i = 1, \, 2, \, \cdots, \, n_d$ respectively.
\begin{problem}
Assuming that $y$ and $u$ are the only available measurements of $\Sigma_p$ and $x_p = \0$ at $t = 0$, synthesize a barrier pair $(\mathbf{B}, \, \Sigma_k)$ and a switching system $\Sigma_s$ that satisfy the following requirements.
\begin{itemize}
\item[(1)] $(\mathbf{B}, \, \Sigma_k)$ satisfies conditions (a) and (b) defined in Definition~\ref{def:bp} for the uncertain system $\Sigma_p$.
\item[(2)] $\Sigma_s$ switches the input of $\Sigma_p$ from its original value $u = \hat{u}$ to $u = \Sigma_k(y)$ when the barrier function $\mathbf{B}$ approaches $1$. Conversely, it switches the input back to $u = \hat{u}$ when $\mathbf{B}$ deviates from $1$.
\end{itemize}
\end{problem}

Fig.~\ref{fig:problem} shows a block diagram of $\Sigma_s$ that is similar to the switching systems proposed in \cite{wieland2007constructive, thomas2018safety}. 
This type of switching system seeks to maintain the safety of $\Sigma_p$ through a minimum intervention in its original input $u = \hat{u}$. 
However, unlike the previous works, the full state of $\Sigma_p$ is not available in our problem. 
Therefore, our switching system cannot switch between $u = \hat{u}$ and $u = \Sigma_k (y)$ based on the exact value of our barrier function $\mathbf{B}$. 

To address this limitation, we will solve our problem in two steps. 
In Sec.~\ref{sec:safety-control-0}, we show a synthesis method that creates a barrier pair for $\Sigma_p$. 
In Sec.~\ref{sec:fault-detection-0}, we propose a robust barrier function estimator that finds an upper bound for $\mathbf{B}$ using only the measurements of $y$ and $u$. 
Based on this upper bound for $\mathbf{B}$, we will construct $\Sigma_s$ in Sec.~\ref{sec:switching}.

\section{Barrier Pair Synthesis} \label{sec:safety-control-0}

In this section, we focus on the barrier pair synthesis for our system $\Sigma_p$ described in \eqref{eq:model-0}. 
First we formulate conditions (a) and (b) in Definition~\ref{def:bp} as linear matrix inequality (LMI) constraints. 
Then we introduce our LMI optimization problems for barrier pair synthesis.

\subsection{LMIs for Invariance Property} \label{sec:invariance}

Let us first consider the case in which there already exists a dynamic output feedback controller $\Sigma_k$ in the form of \eqref{eq:controller-siso}. 
In Proposition~\ref{prop:proposition-1}, we will explain how to find a barrier function $\mathbf{B}$ such that $(\mathbf{B}, \, \Sigma_k)$ satisfies condition (a) in Definition~\ref{def:bp}.

\begin{proposition} \label{prop:proposition-1}
Suppose that $\Sigma_k$ is a controller in the form of \eqref{eq:controller-siso}.
If there exist $P \succ \0$, $\upmu_{d(1)}, \ \upmu_{d(2)}, \ \dots, \ \upmu_{d(n_d)} \geq 0$ and $\upmu_{p(1)}, \ \upmu_{p(2)}, \ \dots, \ \upmu_{p(n_{\uptheta})} > 0$ such that
\begin{equation} \label{eq:LMI-lyapunov-0}
\begin{bmatrix*}
A_\mathtt{CL}^\top P + P A_\mathtt{CL} + \frac{\upmu_\mathtt{CL}}{\upvarepsilon^2} P & \star & \star \\
B_\mathtt{CL}^\top P & -\diag \{ M_d, M_p \} & \star \\
M_{p} C_\mathtt{CL} & M_{p} D_\mathtt{CL} & - M_{p}
\end{bmatrix*}
\prec \mathbf{0},
\end{equation}
where 
$M_d \triangleq \diag \{ \upmu_{d(1)}, \ \upmu_{d(2)}, \ \dots, \ \upmu_{d(n_d)} \}$,
$M_p \triangleq \diag \{ \upmu_{p(1)}, \ \upmu_{p(2)}, \ \dots, \ \upmu_{p(n_{\uptheta})} \}$,
and
$\upmu_\mathtt{CL} = \sum_{i} ^ {n_d} \upmu_{d(i)} \bar{d}_{(i)} ^ 2$,
then
\begin{equation} \label{eq:barrier-function}
\mathbf{B} (x_\mathtt{CL}) \triangleq \norm{x_\mathtt{CL}}_{P},
\end{equation}
is a barrier function
and
$( \mathbf{B} (x_\mathtt{CL}), \, \Sigma_k )$ satisfies condition~(a) in Definition~\ref{def:bp} for a system $\Sigma_p$ described in \eqref{eq:model-0} with $d \in \mathcal{D}$.
\end{proposition}

\begin{proof}
First, let us consider $\mathbf{B}^2(x_\mathtt{CL}) = x_\mathtt{CL}^\top P x_\mathtt{CL}$ as a candidate quadratic Lyapunov function for our closed-loop system with $u = \Sigma_k(y)$. 
The condition for $(\mathbf{B}(x_\mathtt{CL}), \Sigma_k)$ to satisfy condition (a) in Definition~\ref{def:bp} for $\Sigma_p$ in \eqref{eq:model-0} under the assumption $d \in \mathcal{D}$ can be expressed as $\frac{\df \mathbf{B}^2(x_\mathtt{CL})}{\df t} < 0$, or equivalently
\begin{equation} \, \label{eq:B-dot}
\begin{bmatrix}
x_\mathtt{CL} \\
d \\
p
\end{bmatrix}
^ \top 
\begingroup
\renewcommand*{\arraystretch}{1.6}
\begin{bmatrix*}
A_\mathtt{CL} ^ \top P + P A_\mathtt{CL} & \star \\
B_\mathtt{CL}^{\top} P & \0
\end{bmatrix*}
\endgroup
\begin{bmatrix}
x_\mathtt{CL} \\
d \\
p
\end{bmatrix}
< 0,
\end{equation}
for all $x_\mathtt{CL}$, $d$, and $p$ that
\begin{equation} \label{eq:s-condition}
\begin{aligned}
x_\mathtt{CL} ^ \top P x_\mathtt{CL} \geq \upvarepsilon ^ 2,
\quad
& 
d_{(i)} ^ {2}
\leq \bar{d}_{(i)} ^ 2,
&&
\text{for} \ i = 1, \, 2, \, \cdots, \, n_d,
\\
&
p_{(i)} ^ 2
\leq q_{(i)} ^ 2,
\quad
&&
\text{for} \ i = 1, \, 2, \, \cdots, \, n_{\uptheta},
\end{aligned} 
\end{equation}
where $q_{(i)}$ for $i = 1, \, 2, \, \cdots, \, n_{\uptheta}$ is in terms of $x_\mathtt{CL}$, $d$, and $p$ as defined in \eqref{eq:model-CL}.
Using the S-procedure, \eqref{eq:B-dot} holds under the conditions in \eqref{eq:s-condition} if there exist $\upmu_{d(1)}, \ \upmu_{d(2)}, \ \dots, \ \upmu_{d(n_d)} \geq 0$ and $\upmu_{p(1)}, \ \upmu_{p(2)}, \ \dots, \ \upmu_{p(n_{\uptheta})} \geq 0$ such that for all $x_\mathtt{CL}$, $d$, and $p$
\begin{equation} \label{eq:LMI-lyapunov-1}
\begin{aligned}
\begingroup
\begin{bmatrix}
x_\mathtt{CL} \\
d \\
p
\end{bmatrix}
^ {\hspace{-3pt} \top}
\endgroup
\hspace{-3pt}
\Omega_{P}
\begin{bmatrix}
x_\mathtt{CL} \\
d \\
p
\end{bmatrix}
+ \sum_{i} ^ {n_d} \upmu_{d(i)} \bar{d}_{(i)} ^ 2
- \upmu_\mathtt{CL}
< 0
\end{aligned}
\end{equation}
where
\begin{equation} \label{eq:HP}
\begin{aligned}
\Omega_{P}
&
\triangleq
\begin{bmatrix*}
A_\mathtt{CL} ^ \top P + P A_\mathtt{CL} + \frac{\upmu_\mathtt{CL}}{\upvarepsilon ^ 2} P & \star \\
B_\mathtt{CL}^{\top} P & -\diag \{ M_d, M_p \}
\end{bmatrix*}
\\
&
+
\begin{bmatrix*}
C_\mathtt{CL} & D_\mathtt{CL}
\end{bmatrix*}
^ \top
M_p
\begin{bmatrix*}
C_\mathtt{CL} & D_\mathtt{CL}
\end{bmatrix*}
.
\end{aligned}
\end{equation} 
To ensure \eqref{eq:LMI-lyapunov-1}, we need $\Omega_{P} \prec \mathbf{0}$ and $\upmu_\mathtt{CL} = \sum_{i} ^ {n_d} \upmu_{d(i)} \bar{d}_{(i)} ^ 2$. 
If $\upmu_{p(1)}, \ \upmu_{p(2)}, \ \dots, \ \upmu_{p(n_{\uptheta})} > 0$, then $M_p \succ \mathbf{0}$ and we can use the Schur complement to show that $\Omega_{P} \prec \mathbf{0}$ holds if and only if \eqref{eq:LMI-lyapunov-0} holds. 
Therefore, $(\mathbf{B}(x_\mathtt{CL}), \Sigma_k)$ satisfies condition (a) in Definition~\ref{def:bp} for $d \in \mathcal{D}$ if there exist $\upmu_{d(1)}, \ \upmu_{d(2)}, \ \dots, \ \upmu_{d(n_d)} \geq 0$ and $\upmu_{p(1)}, \ \upmu_{p(2)}, \ \dots, \ \upmu_{p(n_{\uptheta})} > 0$ such that \eqref{eq:LMI-lyapunov-0} holds.
\end{proof}

Proposition~\ref{prop:proposition-1} enforces $x_\mathtt{CL}$ to converge to a residual set $\{ x_\mathtt{CL} : \mathbf{B} (x_\mathtt{CL}) \leq \upvarepsilon \}$ under a bounded disturbance $d \in \mathcal{D}$.
However, if $\Sigma_k$ is undetermined, \eqref{eq:LMI-lyapunov-0} will not be an LMI.
Especially in the upper left block of \eqref{eq:LMI-lyapunov-0}, the controller parameters of $\Sigma_k$ and variable $P$ appear in a nonlinear form.
We will resolve this issue in Proposition~\ref{prop:proposition-2} below.

Let us define $Q \triangleq P ^ {- 1}$ and partition $P$ and $Q$ as
\begin{equation} \label{eq:PQ}
P = 
\begin{bmatrix*}[l] 
X & V \\ 
V ^ \top & \star 
\end{bmatrix*}
\quad \text{and} \quad
Q =
\begin{bmatrix*}[l] 
Y & W \\ 
W ^ \top & \star 
\end{bmatrix*}
,
\end{equation}
where $\star$ in the lower-right blocks of matrices $P$ and $Q$ represents unspecified elements with negligible influence in our controller synthesis.
In addition, we define
\begin{equation}
\Pi_1 \triangleq 
\begin{bmatrix*}[l]
 \I & X \\
\0 & V ^ \top
\end{bmatrix*}
\quad \text{and}
\quad
\Pi_2 \triangleq 
\begin{bmatrix*}[l]
Y        & \I \\
W ^ \top & \0
\end{bmatrix*}
.
\end{equation}
Notice that $\Pi_1 = P \Pi_2$ and $\Pi_2 = Q \Pi_1$ are satisfied by construction. 
By performing a congruence transformation with $\diag ( \Pi_2, \ \I, \ \I )$ on \eqref{eq:LMI-lyapunov-0} \cite{scherer1997multiobjective}, we obtain 
\begin{equation} 
\begin{bmatrix*}
\Omega_{A}        & \star     & \star  \\
\Omega_{B} ^ {\top} &   - \diag \{ M_d, M_p \}     & \star  \\
M_p \Omega_{C}    &     M_p D_\mathtt{CL} &     - M_{p} 
\end{bmatrix*}
\prec 0, 
\label{eq:LMI-lyapunov}
\end{equation}
where
\begin{flalign}
\Omega_{A} & \triangleq
\begingroup
\setlength\arraycolsep{2.0pt}
\mathtt{He} 
\bigg\{
\begin{bmatrix*}
A_{p} Y + A_{u} G & A_{p} \\
  E & X A_{p} + F A_{c} \bar{C}
\end{bmatrix*}
\bigg\}
\endgroup
+ 
\frac{\upmu_\mathtt{CL}}{\upvarepsilon ^ 2}
\begingroup
\setlength\arraycolsep{2.0pt}
\begin{bmatrix*}
 Y & \I \\ 
\I &  X
\end{bmatrix*}
\endgroup
,
\notag 
\\
\Omega_{B} & \triangleq
\text{\normalsize $
\begin{bmatrix*}[r]
\I & \0 &   B_{h} & A_{y} B_{c} \\
X  &  F & X B_{h} & F B_{c} + X A_{y} B_{c}
\end{bmatrix*}
$}, \label{eq:LMI-lyapunov-variable}
\\
\Omega_{C} & \triangleq
\text{\normalsize $
\begin{bmatrix*}[r]
C_{u} G + C_{y} A_{c} \bar{C} Y & C_{y} A_{c} \bar{C} \\
C_{c} \bar{C} Y & C_{c} \bar{C}
\end{bmatrix*}
$}
, \notag \\
\text{and} && \notag \\
E & \triangleq V A_k W ^ \top + F A_c \bar{C} Y + X A_u G + X A_p Y, \label{eq:variable-1} \\
F & \triangleq V B_k, \quad G \triangleq C_k W ^ \top \notag.
\end{flalign}
\begin{proposition} \label{prop:proposition-2}
Supposing that $d \in \mathcal{D}$, there exist a barrier function $\mathbf{B} (x_\mathtt{CL}) = \norm{x_\mathtt{CL}}_{P}$ and a controller $\Sigma_k$ in the form of \eqref{eq:controller-siso} such that $( \mathbf{B} (x_\mathtt{CL}), \, \Sigma_k)$ satisfies condition~(a) in Definition~\ref{def:bp} for a system $\Sigma_p$ in the form of \eqref{eq:model-0}, if there exist $\upmu_{d(1)}, \ \upmu_{d(2)}, \ \dots, \ \upmu_{d(n_d)} \geq 0$ and $\upmu_{p(1)}, \ \upmu_{p(2)}, \ \dots, \ \upmu_{p(n_{\uptheta})} > 0$ such that \eqref{eq:LMI-lyapunov} holds and
\begin{equation} \label{eq:z1-z2}
\begin{bmatrix} 
Y  & \I \\ 
\I & X 
\end{bmatrix} 
\succ 0. 
\end{equation}
\end{proposition}

\begin{proof}
According to Proposition~\ref{prop:proposition-1}, there exists $( \mathbf{B} (x_\mathtt{CL}), \, \Sigma_k)$ that satisfies condition~(a) in Definition~\ref{def:bp} for $\Sigma_p$ in \eqref{eq:model-0} under the assumption $d \in \mathcal{D}$ if there exist $\upmu_{d(1)}, \ \upmu_{d(2)}, \ \dots, \ \upmu_{d(n_d)} \geq 0$ and $\upmu_{p(1)}, \ \upmu_{p(2)}, \ \dots, \ \upmu_{p(n_{\uptheta})} > 0$ such that \eqref{eq:LMI-lyapunov-0} holds. 
Since \eqref{eq:LMI-lyapunov-0} is obtained by performing a congruence transformation with $\diag ( \Pi_1, \ \I, \ \I ) $ on \eqref{eq:LMI-lyapunov}, \eqref{eq:LMI-lyapunov-0} holds if \eqref{eq:LMI-lyapunov} holds and $\diag ( \Pi_1, \ \I, \ \I ) $ is non-singular. 

Now, let us focus on the existence of non-singular $\diag ( \Pi_1, \ \I, \ \I ) $. Obviously, $\diag ( \Pi_1, \ \I, \ \I ) $ is non-singular if and only if $V$ is non-singular. Since $P Q = \I$, we have $V W ^ \top = \I - X Y$. 
Therefore, there exists non-singular $V$ and $W$ if $\I - X Y$ is non-singular. $\I - X Y$ is non-singular if $Y \succ \0$ and $X - Y ^ {- 1} \succ \0$. 
Using the Schur complement, $Y \succ \0$ and $X - Y ^ {- 1} \succ \0$ if \eqref{eq:z1-z2} holds. 
Hence, we conclude that there exists non-singular $\diag ( \Pi_1, \ \I, \ \I ) $ if \eqref{eq:z1-z2} holds.

To summarize, the existence of $(\mathbf{B}(x_\mathtt{CL}), \Sigma_k)$ satisfying condition~(a) in Definition~\ref{def:bp} for $\Sigma_p$ in \eqref{eq:model-0} is guaranteed under the assumption $d \in \mathcal{D}$, provided that $\upmu_{d(1)}, \ \upmu_{d(2)}, \ \dots, \ \upmu_{d(n_d)} \geq 0$ and $\upmu_{p(1)}, \ \upmu_{p(2)}, \ \dots, \ \upmu_{p(n_{\uptheta})} > 0$ can be found such that \eqref{eq:LMI-lyapunov-0} holds. 
The validity of \eqref{eq:LMI-lyapunov-0} can be ensured by satisfying the conditions given by \eqref{eq:LMI-lyapunov} and \eqref{eq:z1-z2}.
\end{proof}

Notice that \eqref{eq:LMI-lyapunov} is not an LMI in $X$, $Y$, $E$, $F$, $G$, $M_d$, and $M_p$.
However, if we define the values of $\upmu_\mathtt{CL}, \, \upmu_{p(1)}, \ \upmu_{p(2)}, \ \dots, \ \upmu_{p(n_{\uptheta})}$ a priori, \eqref{eq:LMI-lyapunov} becomes an LMI in $( X, \ Y, \ E, \ F, \ G )$ and $M_d$.

\subsection{LMIs for State and Input Limits}

Now, let us consider condition (b) in Definition~\ref{def:bp} for barrier function $\mathbf{B} (x_\mathtt{CL}) = \norm{x_\mathtt{CL}}_{P}$ and safety controller $\Sigma_k$ in the form of \eqref{eq:controller-siso}.
Since $f_{i} ^ \top x_p = f_{i} ^ \top S_p x_\mathsf{CL}$ and $u_{(i)} = C_k ^ {(i)} S_k x_\mathsf{CL}$ in the expressions of $\mathcal{X}_s$ and $\mathcal{U}$ in \eqref{eq:bars}, condition~(b) in Definition~\ref{def:bp} holds for $( \mathbf{B} (x_\mathtt{CL}), \, \Sigma_k)$ if
\begin{alignat}{5}
& 
f_{i} ^ \top 
&& 
\overbracket{
S_p Q S_p ^ \top
}^{Y} 
&& 
f_{i} 
&& 
\leq 
1, \quad 
&&
\text{for} \ i = 1, \, 2, \, \dots, \, n_s, \label{eq:LMI-xb} \\
& 
C_k ^ {(i)} 
&& 
S_k Q S_k ^ \top
&& 
C_k ^ {(i) \top} 
&& 
\leq 
\bar{u}_{(i)} ^ 2, \quad 
&&
\text{for} \ i = 1, \, 2, \, \dots, \, n_u, \label{eq:LMI-ub-0}
\end{alignat}
where $C_k ^ {(i)}$ is the $i$-th row of $C_k$.
If $\Sigma_k$ is undetermined, \eqref{eq:LMI-ub-0} will not be LMIs. 
We address this issue as follows.
Since $Q = P ^ {- 1}$, we can use the Schur complement to obtain that \eqref{eq:LMI-ub-0} holds if 
\begin{align}
\begin{bmatrix}P & \star \\ C_k ^ {(i)} S_k & \bar{u}_{(i)} ^ 2 \end{bmatrix}\succeq \0, \quad \text{for} \ i = 1, \, 2, \, \dots, \, n_u. \label{eq:LMI-ub-1}
\end{align}
By performing a congruence transformation with $\diag ( \Pi_2, \ \I) $ on \eqref{eq:LMI-ub-1}, we obtain that \eqref{eq:LMI-ub-0} holds if
\begin{align}
&
\begin{bmatrix}
Y  & \star & \star \\
\I &     X & \star \\
G ^ {(i)} & \0 & \bar{u}_{(i)} ^ 2
\end{bmatrix}
\succeq \0, \quad \text{for} \ i = 1, \, 2, \, \dots, \, n_u, \label{eq:LMI-ub}
\end{align}
where $G ^ {(i)}$ is the $i$-th row of $G$. The LMIs in \eqref{eq:LMI-ub} are in terms of our new variable set $( X, \ Y, \ E, \ F, \ G )$ introduced in Proposition~\ref{prop:proposition-2}. 
In addition, we can use the Schur complement to obtain that \eqref{eq:LMI-ub} also implies condition \eqref{eq:z1-z2} in Proposition~\ref{prop:proposition-2}.

\subsection{Barrier Pair Construction} \label{sec:barrier-pair-construction}

Through a convex optimization
\begin{equation} \label{eq:optimization-volume}
\begin{aligned}
& \underset{X, Y, E, F, G, M_d}{\mathtt{maximize \quad}}
& & \mathtt{log}(\mathtt{det} (Y)) \\
& \mathtt{subject \ to} & & \eqref{eq:LMI-lyapunov}, \ \eqref{eq:LMI-xb}, \ \eqref{eq:LMI-ub}, \\
&&& 
X \succ 0, \ Y \succ 0, \\
&&& 
\upmu_{d(1)}, \ \upmu_{d(2)}, \ \dots, \ \upmu_{d(n_d)} \geq 0, \\
&&&
\upmu_\mathtt{CL} = \sum_{i} ^ {n_d} \upmu_{d(i)} \bar{d}_{(i)} ^ 2,
\end{aligned}
\end{equation}
we obtain a solution of $( X, \ Y, \ E, \ F, \ G )$ that maximize the volume of the $x_p$ space projection $\{ S_p x_\mathtt{CL}: \ \mathbf{B} (x_\mathtt{CL}) \leq 1 \}$ of the unit sub-level set of $\mathbf{B} (x_\mathtt{CL})$. 

There are multiple ways to construct a controller $\Sigma_k$ based on a solution of $( X, \ Y, \ E, \ F, \ G )$.
For example, \cite[Lemma~7.9]{dullerud2013course} defines $V V ^ \top = X - Y ^ {- 1}$ and $W = - Y V$.
After obtaining $V$ and $W$, we can construct our controller parameters $A_k$, $B_k$, and $C_k$ according to \eqref{eq:variable-1}.
 
The barrier pair synthesis in optimization \eqref{eq:optimization-volume} is greatly influenced by the scalar parameter $\upvarepsilon$ and scalar variables $\upmu_\mathtt{CL}, \, \upmu_{p(1)}, \ \upmu_{p(2)}, \ \dots, \ \upmu_{p(n_{\uptheta})}$. 
Specifically, a smaller value of $\upvarepsilon$ has a notable effect on disturbance rejection, resulting in a more conservative controller $\Sigma_k$. 
Consequently, the value of $\mathtt{log}(\mathtt{det} (Y))$ and the size of $\{ S_p x_\mathtt{CL}: \mathbf{B} (x_\mathtt{CL}) \leq 1 \}$ become relatively smaller. The selection of scalar variables $\upmu_\mathtt{CL}, \, \upmu_{p(1)}, \ \upmu_{p(2)}, \ \dots, \ \upmu_{p(n_{\uptheta})}$ also plays an important role in determining the solution to optimization \eqref{eq:optimization-volume}. 
However, comprehending their specific effects on optimization results may be less intuitive compared to the influence of $\upvarepsilon$.

Moreover, the scalar variables $\upmu_{p(1)}, \ \upmu_{p(2)}, \ \dots, \ \upmu_{p(n_{\uptheta})}$ address the model uncertainties in the DNLDI model. 
To fine-tune these variables, the D-K iteration method is commonly employed \cite{balas1993mu}. 
This iterative approach alternatively adjusts the scalar variables and controllers in the LMI synthesis process to find an optimal solution. 
Recently, an algorithm in \cite{chen2020alternative} extends the idea of the D-K iteration method to the LMI synthesis problems using the parameter transformation scheme in \cite{scherer1997multiobjective}, providing us with a valuable technique in this context.

\subsection{Barrier Function Composition}

Both the safety controller $\Sigma_k$ and a barrier function $\mathbf{B}_0 (x_\mathtt{CL}) = \norm{x_\mathtt{CL}}_{P_0}$ can be constructed from the solution $( X_0, \ Y_0, \ E_0, \ F_0, \ G_0 )$ to the optimization problem in \eqref{eq:optimization-volume}. 
Based on this safety controller $\Sigma_k$, we can further optimize the shape of the invariant set of our barrier function using the composite Lyapunov function method in \cite{hu2003composite}.

First, let us go back to Proposition~1. 
By performing a congruence transformation with $\diag ( Q, \ \I, \ N_p, \ N_p )$ on \eqref{eq:LMI-lyapunov-0}, we obtain
\eqref{eq:LMI-lyapunov-composite} holds
\begin{equation} \label{eq:LMI-lyapunov-composite}
\begingroup
\renewcommand*{\arraystretch}{1.2}
\begin{bmatrix*}
A_\mathtt{CL} Q + Q A_\mathtt{CL} ^ \top + \frac{\upmu_\mathtt{CL}}{\upvarepsilon ^ 2} Q & \star & \star \\
\diag (\I, N_p) B_\mathtt{CL}^{\top} & - \diag (M_d, N_p)  & \star \\
C_\mathtt{CL} Q & D_\mathtt{CL} & - N_{p} \\
\end{bmatrix*}
\endgroup
\prec \0.
\end{equation}
where $N_p \triangleq \diag \{ \upnu_{p(1)}, \ \upnu_{p(2)}, \ \dots, \ \upnu_{p(n_{\uptheta})} \}$ and $\upnu_{p(i)} \triangleq \frac{1}{\upmu_{p(i)}}$ for $i = 1, \, 2, \, \cdots, \, n_{\uptheta}$.
Since $\Sigma_k$ is known, \eqref{eq:LMI-lyapunov-composite} becomes an LMI in $Q$, $M_d$, and $N_p$ if we fix the value of $\upmu_\mathtt{CL}$.

Then, let us define $Q_i \in \R ^ {2 n \times 2 n}$, $\upxi_i \in \R ^ {n_{i}}$, and $S_{i} \in \R ^ {n_{i} \times 2 n}$ for $i = 1, \, 2, \, \cdots, \, n_{\upgamma}$ and $0 < n_{i} \leq 2 n$. 
Substituting $Q = Q_i$ into \eqref{eq:LMI-xb}, \eqref{eq:LMI-ub-0}, and \eqref{eq:LMI-lyapunov-composite}, we construct $n_{\upgamma}$ convex optimizations
\begin{equation} \label{eq:optimization-direction-i}
\begin{aligned}
& \underset{Q_i, M_d, N_p}{\mathtt{minimize}}
& & \uprho_i \\
& \mathtt{subject \ to} & & 
\eqref{eq:LMI-xb}, \ \eqref{eq:LMI-ub-0}, \ \eqref{eq:LMI-lyapunov-composite}, \\
&&&
Q_i \succ 0, \quad
\begin{bmatrix}
\uprho_i & \star \\
\upxi_i  & S_{i} Q_i S_{i} ^ \top
\end{bmatrix}
\succeq \0, \\
&&&  
\upmu_{d(1)}, \ \upmu_{d(2)}, \ \dots, \ \upmu_{d(n_d)} \geq 0, \\
&&&  
\upnu_{p(1)}, \ \upnu_{p(2)}, \ \dots, \ \upnu_{p(n_{\uptheta})} > 0, \\ 
&&& 
\upmu_\mathtt{CL} = \sum_{i} ^ {n_d} \upmu_{d(i)} \bar{d}_{(i)} ^ 2,
\end{aligned}
\end{equation}
to synthesize $n_{\upgamma}$ additional barrier functions $\mathbf{B}_i (x_\mathtt{CL}) \triangleq \norm{x_\mathtt{CL}}_{P_i}$ for $P_i \triangleq Q_i ^ {- 1}$.
Notice that $\begin{bmatrix} \uprho_i & \star \\ \upxi_i  & S_{i} Q_i S_{i} ^ \top \end{bmatrix} \succeq \0$ is equivalent to $\upxi_i ^ \top (S_{i} Q_i S_{i} ^ \top) ^ {- 1} \upxi_i \leq \uprho_i$ by the Schur complement.
If we consider $\frac{\upxi_i}{\sqrt{\uprho_i}}$ as a point in the $S_i x_\mathtt{CL}$ space, it lies within the ellipsoidal region defined by $\{ S_i x_\mathtt{CL} : \mathbf{B}_i (x_\mathtt{CL}) \leq 1 \}$.
By minimizing $\uprho_i$ in \eqref{eq:optimization-direction-i}, we maximize the radius of $\{ S_i x_\mathtt{CL} : \mathbf{B}_i (x_\mathtt{CL}) \leq 1 \}$ in the direction of vector $\upxi_i$.

Based on our original barrier function $\mathbf{B}_0 (x_\mathtt{CL})$ and all the new barrier functions synthesized through optimization \eqref{eq:optimization-direction-i}, we can introduce our composite barrier function.
Since $Q_1, \, Q_2, \, \dots, \, Q_{n_{\upgamma}}$ all satisfy LMIs \eqref{eq:LMI-xb}, \eqref{eq:LMI-ub-0}, and \eqref{eq:LMI-lyapunov-composite},
\begin{equation} \label{eq:variable-gamma}
Q_{\upgamma} \triangleq 
\sum_{i = 0} ^ {n_{\upgamma}} \upgamma_i Q_i
\end{equation}
also satisfies LMIs \eqref{eq:LMI-xb}, \eqref{eq:LMI-ub-0}, and \eqref{eq:LMI-lyapunov-composite},
where
\begin{flalign}
\upgamma & \triangleq 
\begin{bmatrix}
\upgamma_0 & \upgamma_1 & \cdots & \upgamma_{n_{\upgamma}}
\end{bmatrix}
^ \top 
\in \Gamma
&&
\label{eq:gamma}
\\
\text{and}
&&
\notag
\\
\Gamma & \triangleq \left\{ \upgamma : 
\begin{matrix}
\sum_{i = 0} ^ {n_{\upgamma}} \upgamma_i = 1
\end{matrix}
, \ \upgamma_i \geq 0, \ i = 0, \, 1, \, \dots, \, n_{\upgamma} \right\}.
&&
\label{eq:Gamma}
\end{flalign}
\begin{lemma} \label{lemma:Bc}
Let us define
\begin{equation} \label{eq:Bc}
\mathbf{B}_\mathtt{C} (x_\mathtt{CL}) \triangleq \underset{\upgamma \in \Gamma}{\min} \ \mathbf{B}_{\upgamma} (x_\mathtt{CL}).
\end{equation}
where 
$\mathbf{B}_{\upgamma} (x_\mathtt{CL}) \triangleq \norm{x_\mathtt{CL}}_{P_{\upgamma}}$ for $P_{\upgamma} \triangleq Q_{\upgamma} ^ {- 1}$. Then, $\mathbf{B}_\mathtt{C} (x_\mathtt{CL})$ satisfies the following conditions:
\begin{itemize}
\item[(a)] Let $\upgamma_\star$ be the optimal $\upgamma$ such that $\mathbf{B}_{\upgamma_\star} (x_\mathtt{CL}) \triangleq \underset{\upgamma \in \Gamma}{\min} \ \mathbf{B}_{\upgamma} (x_\mathtt{CL})$, then $\dfrac{\uppartial \mathbf{B}_\mathtt{C} ^ 2}{\uppartial x_\mathtt{CL}} = 2 Q_{\upgamma_\star} ^ {- 1} x_\mathtt{CL}$.
\vspace{2pt}
\item[(b)] $\{ x_\mathtt{CL} : \mathbf{B}_\mathtt{C} (x_\mathtt{CL}) \leq 1 \}$ is the convex hull of $\{ x_\mathtt{CL} : \mathbf{B}_i (x_\mathtt{CL}) \leq 1 \}$ for $i = 0, \, 1, \, \cdots, \, n_{\upgamma}$.
\end{itemize}
\end{lemma}
\begin{proof}
See \cite[Appendix~A]{hu2003composite}.
\end{proof}
In this paper, we refer to $\mathbf{B}_\mathtt{C} (x_\mathtt{CL})$ in \eqref{eq:Bc} as a composite barrier function.
Since $Q_{\upgamma ^ \star}$ satisfies LMI \eqref{eq:LMI-lyapunov-composite} and $\frac{\uppartial \mathbf{B}_\mathtt{C} ^ 2}{\uppartial x_\mathtt{CL}} = 2 Q_{\upgamma_\star} ^ {- 1} x_\mathtt{CL}$ from condition (a) in Lemma~\ref{lemma:Bc}, we have $\frac{\df \mathbf{B}_\mathtt{C} ^2}{\df t} < 0$ for $d \in \mathcal{D}$.
In addition, condition (b) in Lemma~\ref{lemma:Bc} implies that $(\mathbf{B}_\mathtt{C}, \, \Sigma_k)$ satisfies condition (b) in Definition~\ref{def:bp}.
Therefore, $(\mathbf{B}_\mathtt{C}, \, \Sigma_k)$ is also a barrier pair, which satisfies conditions (a) and (b) in Definition~\ref{def:bp}.

In Fig.~\ref{fig:result-mimo}.(a)-(b), we show an example of $\{ x_\mathtt{CL} : \mathbf{B}_\mathtt{C} (x_\mathtt{CL}) \leq 1 \}$ for $n_{\upgamma} = 6$, in which our composite barrier function approach leverages two optimization processes to enhance the safety guarantees of the system. 
First, optimization \eqref{eq:optimization-volume} is designed to maximize the overall volume of the invariant set, resulting in a broader safety region encompassing all directions. 
Conversely, optimization \eqref{eq:optimization-direction-i} specifically focuses on maximizing the invariant set of the barrier function in a particular direction $\upxi_i$, utilizing the same controller $\Sigma_k$ derived from \eqref{eq:optimization-volume}. 
By combining the outcomes of these optimizations through a composition barrier function, we extend the invariant set of the barrier function obtained from \eqref{eq:optimization-volume} in each specific direction, thereby expanding the safety region along those directions.

\section{Barrier Function Estimation} \label{sec:fault-detection-0}

Notice that the operation of the switching system $\Sigma_s$ in Fig.~\ref{fig:problem} requires the real-time evaluation of the barrier function $\mathbf{B}_\mathtt{C}$. 
However, this is challenging since the state $x_p$ of our plant $\Sigma_p$ is not available.
In this section, we will introduce a robust barrier function estimator that provides an upper bound for $\mathbf{B}_\mathtt{C}$ using only the measurements of $y$ and $u$.

\subsection{Identifier-Based Estimator} \label{sec:identifier}

In \cite{morse1980global}, the concept of an identifier-based estimator was developed for the purpose of model identification and adaptive control. 
However, as a byproduct, it also provide us with a robust state estimate $\hat{x}_p$ under the presence of model uncertainty in $\Sigma_p$. 
Let us define 
\begin{equation} \label{eq:Az}
\begin{aligned}
& 
A_{z} \triangleq \diag \big\{ A_{z}^{(1)}, \, A_{z}^{(2)}, \, \cdots, \, A_{z}^{(n_y)} \big\} \\ 
&
B_{z} \triangleq \diag \big\{ B_{z}^{(1)}, \, B_{z}^{(2)}, \, \cdots, \, B_{z}^{(n_y)} \big\} 
\end{aligned}
\end{equation}
where $A_{z}^{(i)} \triangleq \bar{A}^{(i)} - B_{z}^{(i)} \bar{C}^{(i)}$ and $B_{z}^{(i)} \in \R ^ {n_i \times 1}$ are defined by the user such that $A_{z}^{(i)}$ is a Hurwitz matrix for $i = 1, \, 2, \, \cdots, \, n_y$. 
Based on the definition of $A_{z}$ and $B_{z}$ in \eqref{eq:Az}, we have $A_{z} = \bar{A}- B_{z} \bar{C}$. Substituting $\bar{A} = A_{z} + B_{z} \bar{C}$ and \eqref{eq:ths-qip-1} into \eqref{eq:model}, the state-space model of $\Sigma_{p}$ becomes
\begin{alignat}{4}
&
\dot{x}_{p} = 
&
A_{z} x_{p} 
&&
+ 
\overbracket{
\begin{bmatrix} 
A_{u} & A_{y} + B_{z} A_{a} 
\end{bmatrix} 
}^{\hat{A}_{h}}
h
&
+ 
\overbracket{
\setlength\arraycolsep{2.0pt}
\begin{bmatrix} 
\ \ B_{h} & \ \ B_{z} B_{a} 
\end{bmatrix}
}^{\hat{B}_{h}}
\hat{p}_{h}
\notag
\\
&
&
&&
+ 
\underbracket{
\begin{bmatrix} 
\I & - A_{y} - B_{z} A_{a}
\end{bmatrix} 
}_{\hat{A}_{d}}
d 
&
+ 
\underbracket{
\setlength\arraycolsep{2.0pt}
\begin{bmatrix} 
- B_{h} & - B_{z} B_{a} 
\end{bmatrix}
}_{\hat{B}_{d}}
\hat{p}_{d}
\notag
\\
&
\hat{q}_{h} = 
&
\underbracket{
\setlength\arraycolsep{2.0pt}
\begin{bmatrix} 
C_{u} & C_{y} \\ 
   \0 & C_{a} 
\end{bmatrix} 
}_{\hat{C}_{h}}
&&
h
,
\quad
\hat{p}_{h} = 
\begingroup
\setlength\arraycolsep{2.0pt}
\begin{bmatrix} 
\Delta_{h} & \0 \\ 
\0 & \Delta_{a} 
\end{bmatrix}
\endgroup
\hat{q}_{h} 
&
\label{eq:model-siso-new}
\\
&
\hat{q}_{d} = 
&
\underbracket{
\begin{bmatrix} 
\0 & C_{y} \\ 
\0 & C_{a} 
\end{bmatrix} 
}_{\hat{C}_{d}}
&&
d 
,
\quad
\hat{p}_{d} = 
\begingroup
\setlength\arraycolsep{2.0pt}
\begin{bmatrix} 
\Delta_{h} & \0 \\ 
\0 & \Delta_{a} 
\end{bmatrix}
\endgroup
\hat{q}_{d} 
&
\notag
\end{alignat}
where
$\hat{p}_{h}, \, \hat{p}_{d}, \, \hat{q}_{h}, \, \hat{q}_{d} \in \R ^ {n_{\uptheta}}$ are the vectors of disturbance and output signals related to the parameter uncertainty of $\Theta_{h}$ and $\Theta_{a}$.
The identifier-based estimator for system $\Sigma_{p}$ comprises $n_y \cdot n_h$ sensitivity function filters. For each partition $\Sigma_{p}^{(i)}$ of $\Sigma_{p}$, we define $n_h$ sensitivity function filters as follows:
\begin{equation} \label{eq:estimator-h}
\begin{aligned}
\dot{z}_{h(j)}^{(i)} & = 
A_{z}^{(i)\top}
z_{h(j)}^{(i)} + 
\bar{C}^{(i)\top}
h_{(j)}, \quad 
\text{for} \ j = 1, \, 2, \, \dots, \, n_{h},
\end{aligned}
\end{equation}
where $i = 1, \, 2, \, \dots, \, n_y$.
To explain how the sensitivity function filters in \eqref{eq:estimator-h} contribute to robust barrier function estimation, it is essential to incorporate the following filters in the intermediate stages of our discourse:
\begin{align}
\dot{z}_{p(j)}^{(i)} & = 
A_{z}^{(i)\top}
z_{p(j)}^{(i)} + 
\bar{C}^{(i)\top}
\hat{p}_{h(j)}, \quad
\text{for} \ j = 1, \, 2, \, \dots, \, n_{\uptheta}, \label{eq:estimator-p} \\
\dot{z}_{q(j)}^{(i)} & = 
A_{z}^{(i)\top}
z_{q(j)}^{(i)} + 
\bar{C}^{(i)\top}
\hat{q}_{h(j)}, \quad
\text{for} \ j = 1, \, 2, \, \dots, \, n_{\uptheta}, \label{eq:estimator-q}
\end{align}
where $i = 1, \, 2, \, \dots, \, n_y$.
Let us define 
\begin{equation} \label{eq:Es-i}
\begin{aligned}
E_{h(j)}^{(i)} & \triangleq \mathcal{C}_{h(j)}^{(i)}  {\mathcal{C}_{z}^{(i)}}^{-1}, \quad\quad \text{for} \ j = 1, \, 2, \, \dots, \, n_{h}
\\
E_{p(j)}^{(i)} & \triangleq \mathcal{C}_{p(j)}^{(i)}  {\mathcal{C}_{z}^{(i)}}^{-1}, \quad\quad \text{for} \ j = 1, \, 2, \, \dots, \, n_{\uptheta}
\\
E_{q(j)}^{(i)} & \triangleq \mathcal{C}_{q(j)}^{(i)}  {\mathcal{C}_{z}^{(i)}}^{-1}, \quad\quad \text{for} \ j = 1, \, 2, \, \dots, \, n_{\uptheta}
\end{aligned}
\end{equation}
where
$\mathcal{C}_{h(j)}^{(i)}$, $\mathcal{C}_{p(j)}^{(i)}$, $\mathcal{C}_{q(j)}^{(i)}$, and $\mathcal{C}_{z}^{(i)}$ are the controllability matrices of $(A_{z}^{(i)\top}, \ z_{h(j)}^{(i)})$, $(A_{z}^{(i)\top}, \ z_{p(j)}^{(i)})$, $(A_{z}^{(i)\top}, \ z_{q(j)}^{(i)})$, and $(A_{z}^{(i)\top}, \ \bar{C}^{(i)\top})$. 
For convenience of expression, we define
\begin{equation} \label{eq:Es}
\begin{aligned}
E_{h(j)} & \triangleq \diag \big\{ E_{h(j)}^{(1)}, \, E_{h(j)}^{(2)}, \, \cdots, \, E_{h(j)}^{(n_y)} \big\},
\\
E_{p(j)} & \triangleq \diag \big\{ E_{p(j)}^{(1)}, \, E_{p(j)}^{(2)}, \, \cdots, \, E_{p(j)}^{(n_y)} \big\},
\\
E_{q(j)} & \triangleq \diag \big\{ E_{q(j)}^{(1)}, \, E_{q(j)}^{(2)}, \, \cdots, \, E_{q(j)}^{(n_y)} \big\}.
\end{aligned}
\end{equation}
The definitions of $E_{h(j)}^{(i)}$, $E_{p(j)}^{(i)}$ and $E_{q(j)}^{(i)}$ in \eqref{eq:Es-i} will help us to express and prove Lemma~\ref{lemma:xs}, which finds a state estimate $\hat{x}_p$ for the plant state $x_p$ using the sensitivity function filers in \eqref{eq:estimator-h}, \eqref{eq:estimator-p}, and \eqref{eq:estimator-q}.
\begin{lemma} \label{lemma:xs}
If we define 
\begin{equation} \label{eq:xs}
\hat{x}_{p} \triangleq
\sum_{j}^{n_{h}} E_{h(j)}^{\top} \hat{A}_{h(j)} + \sum_{j}^{n_{\uptheta}} E_{p(j)}^{\top} \hat{B}_{h(j)}
\end{equation}
then $e \triangleq x_p - \hat{x}_p$ satisfies 
\begin{equation} \label{eq:model-error}
\dot{e} = 
A_{z}
e 
+ 
\underbracket{
\begin{bmatrix} 
\I & - A_{y} - B_{z} A_{a}
\end{bmatrix} 
}_{\hat{A}_{d}}
d 
+ 
\underbracket{
\begin{bmatrix} 
- B_{h} & - B_{z} B_{a} 
\end{bmatrix}
}_{\hat{B}_{d}}
\hat{p}_{d}.
\end{equation}
\end{lemma}

\begin{proof}
Subtracting \eqref{eq:model-error} from \eqref{eq:model-siso-new}, we obtain that 
\begin{equation} \label{eq:model-siso-hat}
\dot{\hat{x}}_{p} = 
A_{z}
\hat{x}_{p} 
+ 
\underbracket{
\begin{bmatrix} 
A_{u} & A_{y} + B_{z} A_{a}
\end{bmatrix} 
}_{\hat{A}_{h}}
h
+ 
\underbracket{
\begin{bmatrix} 
B_{h} &   B_{z} B_{a} 
\end{bmatrix}
}_{\hat{B}_{h}}
\hat{p}_{h}.
\end{equation}
Therefore, \eqref{eq:model-error} holds if \eqref{eq:model-siso-hat} holds.
Let us partition \eqref{eq:model-siso-hat} as
\begin{equation} \label{eq:model-siso-hat-i}
\begin{aligned}
\dot{\hat{x}}_{p}^{(i)} 
& = 
A_{z}^{(i)}
\hat{x}_{p}^{(i)} 
\\
& + 
\underbracket{
\setlength\arraycolsep{3pt}
\begin{bmatrix} 
A_{u}^{(i)} & A_{y}^{(i)} + B_{z}^{(i)} A_{a}^{(i)}
\end{bmatrix} 
}_{\hat{A}_{h}^{(i)}}
h
+ 
\underbracket{
\setlength\arraycolsep{3pt}
\begin{bmatrix} 
B_{h}^{(i)} &   B_{z}^{(i)} B_{a}^{(i)} 
\end{bmatrix}
}_{\hat{B}_{h}^{(i)}}
\hat{p}_{h}
\end{aligned}
\end{equation}
for $i = 1, \, 2, \, \dots, \, n_y$.
Since $E_{h(j)}^{(i)} \triangleq \mathcal{C}_{h(j)}^{(i)}  {\mathcal{C}_{z}^{(i)}}^{-1}$ and $E_{p(j)}^{(i)} \triangleq \mathcal{C}_{p(j)}^{(i)}  {\mathcal{C}_{z}^{(i)}}^{-1}$, we derive from \eqref{eq:estimator-h} and \eqref{eq:estimator-p} that
\begin{equation} \label{L-1-1}
\begin{aligned} 
\dot{E}_{h(j)}^{(i)} & = A_{z}^{(i)\top} E_{h(j)}^{(i)} + h_{(j)} \cdot \I, \quad \text{for} \ j = 1, \, 2, \, \dots, \, n_{h}, \\
\dot{E}_{p(j)}^{(i)} & = A_{z}^{(i)\top} E_{p(j)}^{(i)} + \hat{p}_{h(j)} \cdot \I, \quad \text{for} \ j = 1, \, 2, \, \dots, \, n_{\uptheta}.
\end{aligned}
\end{equation}
Notice that $E_{h(j)}^{(i)} A_{z}^{(i)\top} = A_{z}^{(i)\top} E_{h(j)}^{(i)}$ and $E_{p(j)}^{(i)} A_{z}^{(i)\top} = A_{z}^{(i)\top} E_{p(j)}^{(i)}$. 
By taking the transpose of \eqref{L-1-1}, we obtain 
\begin{equation} \label{L-1-2}
\begin{aligned} 
\dot{E}_{h(j)}^{(i)\top} & = A_{z}^{(i)} E_{h(j)}^{(i)\top} + h_{(j)} \cdot \I, \quad \text{for} \ j = 1, \, 2, \, \dots, \, n_{h}, \\
\dot{E}_{p(j)}^{(i)\top} & = A_{z}^{(i)} E_{p(j)}^{(i)\top} + \hat{p}_{h(j)} \cdot \I, \quad \text{for} \ j = 1, \, 2, \, \dots, \, n_{\uptheta}.
\end{aligned}
\end{equation}
Then, we obtain \eqref{eq:model-siso-hat-i} by substituting \eqref{L-1-2} into the time derivative of 
\begin{equation} \label{L-1-3}
\begin{aligned}
\hat{x}_p ^ {(i)}
=
\sum_{j}^{n_{h}} E_{h(j)}^{(i)\top} \hat{A}_{h(j)}^{(i)} + \sum_{j}^{n_{\uptheta}} E_{p(j)}^{(i)\top} \hat{B}_{h(j)}^{(i)}
\quad
\text{for} \ i = 1, \, \cdots, \, n_y.
\end{aligned}
\end{equation}
which leads to \eqref{eq:xs}.
\end{proof}

In Lemma~\ref{lemma:xs}, $\hat{x}_p$ is affine in the uncertain model parameters of $\Theta_h$ and $\Theta_a$, which are defined in \eqref{eq:ths-qip-1}.
Let us express \eqref{eq:xs} as
\begin{equation} \label{eq:xs-bar}
\begin{aligned}
& \hat{x}_p = 
\bar{x}_p + \tilde{x}_p
\\
& \bar{x}_p \triangleq 
\sum_{j}^{n_{h}} E_{h(j)}^{\top} \hat{A}_{h(j)}
\qquad
\tilde{x}_p  \triangleq 
\sum_{j}^{n_{\uptheta}} E_{p(j)}^{\top} \hat{B}_{h(j)}
\end{aligned}
\end{equation}
where $\bar{x}_{p}$ is considered as a nominal value of plant state estimate based on the nominal model parameters of $\Theta_{h}$ and $\Theta_{a}$.
Since we know the value of $E_{h(j)}$ from the sensitivity filters in \eqref{eq:estimator-h}, $\bar{x}_p$ is available to us.
However, we do not know the exact value of $\tilde{x}_p$ because $\hat{p}_{h}$ and $E_{p(j)}$ are unknown.
In Lemma~\ref{lemma:ellipsoid}, we will demonstrate that the feasible values of $\tilde{x}_p$ lie within a norm-bounded state space region defined by $E_{q(j)}$. 
Since $\hat{q}_{h} = \hat{C}_{h} h$, the values of $E_{q(j)}$ can also be derived from the sensitivity filters in \eqref{eq:estimator-h}.

\begin{lemma} \label{lemma:ellipsoid}
Supposing that $E_{p(j)}$ and $E_{q(j)}$, as defined in \eqref{eq:Es}, are the solutions of
\begin{equation} \label{L-3-1}
\begin{aligned} 
\dot{E}_{p(j)} & = A_{z}^{\top} E_{p(j)} + \hat{p}_{h(j)} \cdot \I, \quad \text{for} \ j = 1, \, 2, \, \cdots, \, n_{\uptheta}, \\
\dot{E}_{q(j)} & = A_{z}^{\top} E_{q(j)} + \hat{q}_{h(j)} \cdot \I, \quad \text{for} \ j = 1, \, 2, \, \cdots, \, n_{\uptheta}.
\end{aligned}
\end{equation}
then $\tilde{x}_p$ in \eqref{eq:xs-bar} satisfies
\begin{equation} \label{eq:ellipsoid-tilde}
\tilde{x}_p ^ \top \tilde{x}_p \leq n_{\uptheta} \cdot \upphi,
\end{equation}
where 
\begin{equation} \label{eq:ellipsoid-variable}
\begin{aligned}
&
\upphi
\triangleq 
\sum\limits_{j}^{n_{\uptheta}}
\hat{B}_{h(j)}^{\top} E_{q(j)} E_{q(j)}^{\top} \hat{B}_{h(j)}.
\end{aligned}
\end{equation}
\end{lemma}

\begin{proof}
Let us define 
\begin{equation}
\begin{aligned}
\upphi_{(j)}
\triangleq 
\hat{B}_{h(j)}^{\top} E_{q(j)} E_{q(j)}^{\top} \hat{B}_{h(j)},
\quad\quad
\tilde{x}_{p(j)} 
=
E_{p(j)}^{\top} \hat{B}_{h(j)},
\\
\text{for} \ j = 1, \, 2, \, \cdots, \, n_{\uptheta}.
\end{aligned}
\end{equation}
Since $E_{p(j)} = \int_{0} ^ {t} \hat{p}_{h(j)} \cdot \ep ^ {A_{z}^{\top} (t - \uptau)} \df \uptau$ and $E_{q(j)} = \int_{0} ^ {t} \hat{q}_{h(j)} \cdot \ep ^ {A_{z}^{\top} (t - \uptau)} \df \uptau$ are the solutions of \eqref{L-3-1}, we obtain
\begin{equation} \label{L-1-4}
\begin{aligned} 
E_{p(j)} E_{p(j)}^{\top} = \int_{0} ^ {t} \int_{0} ^ {t} \hat{p}_{h(j)} \cdot \hat{p}_{h(j)} \cdot \ep ^ {A_{z}^{\top} (t - \uptau_1)} \ep ^ {A_{z} (t - \uptau_2)} \df \uptau_1 \df \uptau_2, \\
E_{q(j)} E_{q(j)}^{\top} = \int_{0} ^ {t} \int_{0} ^ {t} \hat{q}_{h(j)} \cdot \hat{q}_{h(j)} \cdot \ep ^ {A_{z}^{\top} (t - \uptau_1)} \ep ^ {A_{z} (t - \uptau_2)} \df \uptau_1 \df \uptau_2, \\
\text{for} \ j = 1, \, 2, \, \cdots, \, n_{\uptheta}.
\end{aligned}
\end{equation}
where $\hat{p}_{h(j)} = \hat{\updelta}_{(j)} \hat{q}_{h(j)}$ with $\abs{\hat{\updelta}_{(j)}} \leq 1$. Thus, $E_{p(j)} E_{p(j)}^{\top} \preceq E_{q(j)} E_{q(j)}^{\top}$, which leads to $\tilde{x}_{p(j)} ^ \top \tilde{x}_{p(j)} \leq \upphi_{(j)}$ for $j = 1, \, \cdots, \, n_{\uptheta}$.
Through the Schur complement, 
$\begingroup \begin{bmatrix} \upphi_{(j)} & \star \\ \tilde{x}_{p(j)} & \I \end{bmatrix} \endgroup \succeq \0$ for $j = 1, \, \cdots, \, n_{\uptheta}$.
Then, we have 
\begin{equation}
\begingroup
\begin{bmatrix}
\upphi & \star \\
\tilde{x}_{p} & n_{\uptheta} \cdot \I
\end{bmatrix}
=
\sum\limits_{j}^{n_{\uptheta}}
\begin{bmatrix}
\upphi_{(j)} & \star \\
\tilde{x}_{p(j)} & \I
\end{bmatrix}
\endgroup
\succeq \0,
\end{equation}
which leads to \eqref{eq:ellipsoid-tilde}.
\end{proof}

Lemma~\ref{lemma:xs} and \ref{lemma:ellipsoid} help us understand the uncertainty of our state estimate due to the external input $d$ and the unknown plant parameters.
Next, we will show how to find a measurable upper bound for our barrier function $\mathbf{B}_\mathtt{C}$ using Lemma~\ref{lemma:xs} and \ref{lemma:ellipsoid}.

\subsection{Barrier Function Estimation} \label{sec:estimation}

According to \cite{he2020robust}, $\mathbf{B}_\mathtt{C}(\star)$ is a vector norm function
\begin{equation} \label{eq:norms}
\norm{\star}_\mathtt{C} \triangleq \mathbf{B}_\mathtt{C} ( \star ).
\end{equation}
In Proposition~\ref{prop:b-bar}, we will use Lemma~\ref{lemma:xs}, Lemma~\ref{lemma:ellipsoid}, and the triangle inequality of $\norm{\star}_\mathtt{C}$ to find an upper bound $\bar{\mathbf{B}}_\mathtt{C}$ of our barrier function $\mathbf{B}_\mathtt{C} (x_\mathtt{CL})$.

\begin{proposition} \label{prop:b-bar}
Suppose $\Sigma_p$ is a system described in \eqref{eq:model-0}, $\Sigma_e$ is its identifier-based estimator in the form of \eqref{eq:estimator-h} and $(\mathbf{B}_\mathtt{C}, \, \Sigma_k)$ is its barrier pair, where $\mathbf{B}_\mathtt{C}$ is a composite barrier function in the form of \eqref{eq:Bc} and $\Sigma_k$ is a controller in the form of \eqref{eq:controller-siso}.
Let us define
\begin{equation} \label{eq:V-bar}
\bar{\mathbf{B}}_\mathtt{C} \triangleq 
\bar{r}_\mathtt{CL} + \tilde{r}_p + r_{e},
\end{equation}
where $\bar{r}_\mathtt{CL}, \, \tilde{r}_p, \, r_{e} \geq 0$.
Then $\mathbf{B}_\mathtt{C} (x_\mathtt{CL}) \leq \bar{\mathbf{B}}_\mathtt{C}$ is fulfilled if the following conditions hold:
\begin{itemize}
\item[(a)]
\vspace{2pt}
There exist $\upgamma_0, \, \upgamma_1, \, \cdots, \, \upgamma_{n_{\upgamma}} \geq 0$ such that $\sum\limits_{i = 0} ^ {n_{\upgamma}} \upgamma_i = 1$ and 
\begin{flalign}
\begin{bmatrix}
\bar{r}_\mathtt{CL} ^ 2 & \star \\
\bar{x}_\mathtt{CL} & \sum_{i = 0} ^ {n_{\upgamma}} \upgamma_i Q_i
\end{bmatrix}
\succeq \0. \label{eq:lmi-estimator-1} &&
\end{flalign}
\item[(b)]
Assuming $E_{q(j)}$, as defined in \eqref{eq:Es}, is the solution of $\dot{E}_{q(j)} = A_{z}^{\top} E_{q(j)} + \hat{q}_{h(j)} \cdot \I$ for $j = 1, \, 2, \, \cdots, \, n_{\uptheta}$,
\begin{flalign}
\tilde{r}_p ^ 2 \cdot \I - n_{\uptheta} \cdot \upphi \cdot X_0 \succeq \0, &&
\end{flalign}
where $\upphi$ is defined in \eqref{eq:ellipsoid-variable}.
\item[(c)]
Assuming that $d \in \mathcal{D}$, there exist $\hat{\upmu}_e = \sum_{i} ^ {n_d} \hat{\upmu}_{d(i)} \bar{d}_{(i)} ^ 2$, $\hat{\upmu}_{d(1)}, \ \hat{\upmu}_{d(2)}, \ \dots, \ \hat{\upmu}_{d(n_d)} \geq 0$, and $\hat{\upmu}_{p(1)}, \, \hat{\upmu}_{p(2)}, \, \dots, \, \hat{\upmu}_{p(n_{\uptheta})} \geq 0$ such that
\begin{flalign}
\resizebox{0.9\columnwidth}{!}{$
\begingroup
\setlength\arraycolsep{-0pt}
\begin{bmatrix}
A_z ^ \top X_0 + X_0 A_z + \frac{\hat{\upmu}_e}{r_e ^ 2} X_0 & \star & \star \\
\hat{A}_{d} ^ {\top} X_0 & - \hat{M}_{d} + \hat{C}_{d} ^ {\top} \hat{M}_p \hat{C}_{d} & \star \\
\hat{B}_{d} ^ {\top} X_0 & \0 & - \hat{M}_{p}
\end{bmatrix}
\endgroup
\prec \0.
$}  &&
\label{eq:lmi-estimator-3}
\end{flalign}
\end{itemize}
Here, $\bar{x}_\mathtt{CL} \triangleq \begin{bmatrix} \bar{x}_{p} \\ x_{k} \end{bmatrix}$, $\hat{M}_d \triangleq \diag \{ \hat{\upmu}_{d(1)}, \ \hat{\upmu}_{d(2)}, \ \dots, \ \hat{\upmu}_{d(n_{d})} \}$, and $\hat{M}_p \triangleq \diag \{ \hat{\upmu}_{p(1)}, \ \hat{\upmu}_{p(2)}, \ \dots, \ \hat{\upmu}_{p(n_{\uptheta})} \}$.
\end{proposition}

\begin{proof}
Since $x_p = \bar{x}_p + \tilde{x}_p + e$, $x_\mathtt{CL} \triangleq \begin{bmatrix} x_{p} \\ x_{k} \end{bmatrix}$ and $\bar{x}_\mathtt{CL} \triangleq \begin{bmatrix} \bar{x}_{p} \\ x_{k} \end{bmatrix}$, we have $x_\mathtt{CL} = \bar{x}_\mathtt{CL} + S_p ^ \top \tilde{x}_p + S_p ^ \top e$. 
Using the triangle inequality of $\norm{\star}_\mathtt{C}$, we obtain
\begin{equation} \label{eq:proof-3}
\mathbf{B}_\mathtt{C} (x_\mathtt{CL})
=
\norm{x_\mathtt{CL}}_\mathtt{C}
\leq 
\norm{\bar{x}_\mathtt{CL}}_\mathtt{C} + \norm{S_p ^ \top \tilde{x}_p}_\mathtt{C} + \norm{S_p ^ \top e}_\mathtt{C},
\end{equation}
Based on the definition of $\norm{\star}_\mathtt{C} \triangleq \mathbf{B}_\mathtt{C} ( \star )$ in \eqref{eq:Bc},
\begin{alignat}{4}
\norm{S_p ^ \top \tilde{x}_p}_\mathtt{C}
& \leq 
&
\norm{S_p ^ \top \tilde{x}_p}_{P_0}
= 
&&
\norm{\tilde{x}_p}_{X_0},
&
\\
\norm{S_p ^ \top e}_\mathtt{C}
& \leq 
&
\norm{S_p ^ \top e}_{P_0} 
= 
&&
\norm{e}_{X_0}.
&
\end{alignat}
Therefore, $\mathbf{B}_\mathtt{C} (x_\mathtt{CL}) \leq \bar{\mathbf{B}}_\mathtt{C}$ if 
\begin{equation} \label{eq:prop-3-all}
\norm{ \bar{x}_\mathtt{CL} }_\mathtt{C} + \norm{ \tilde{x}_p }_{X_0} + \norm{e}_{X_0} \leq \bar{r}_\mathtt{CL} + \tilde{r}_p + r_{e}.
\end{equation}
In particular, \eqref{eq:prop-3-all} holds if
\begin{align} \label{eq:prop-3-3}
\norm{ \bar{x}_\mathtt{CL} }_\mathtt{C} \leq \bar{r}_\mathtt{CL}, \quad
\norm{ \tilde{x}_p }_{X_0} \leq \tilde{r}_p, \quad
\text{and} \quad
\norm{e}_{X_0} \leq r_{e}.
\end{align}
In the remainder of this proof, we will demonstrate how conditions (a), (b), and (c) in Proposition~\ref{prop:b-bar} guarantee the fulfillment of the inequalities in \eqref{eq:prop-3-3}.

Based on the definition of $\norm{\star}_\mathtt{C} \triangleq \mathbf{B}_\mathtt{C} ( \star )$ in \eqref{eq:Bc}, $\norm{ \bar{x}_\mathtt{CL} }_\mathtt{C} \leq \bar{r}_\mathtt{CL}$ if 
\begin{equation} \label{eq:prop-3-1}
\bar{x}_\mathtt{CL} ^ \top Q_{\upgamma} ^ {- 1} \bar{x}_\mathtt{CL} \leq \bar{r}_\mathtt{CL} ^ 2,
\end{equation}
where $Q_{\upgamma} \triangleq \sum_{i = 0} ^ {n_{\upgamma}} \upgamma_i Q_i$, $\sum_{i = 0} ^ {n_{\upgamma}} \upgamma_i = 1$, and $\upgamma_i \geq 0$ for all $i = 0, \, 1, \, \cdots, \, n_{\upgamma}$. 
Using the Schur complement, we have \eqref{eq:prop-3-1} holds if condition (a) in Proposition~\ref{prop:b-bar} holds.

According to Lemma~\ref{lemma:ellipsoid}, $\tilde{x}_p ^ \top \tilde{x}_p \leq n_{\uptheta} \cdot \upphi$. 
Since $\norm{ \tilde{x}_p }_{X_0} \leq \tilde{r}_p$ is equivalent to $\tilde{x}_p ^ \top X_0 \tilde{x}_p \leq \tilde{r}_p ^ 2$, we have $\norm{ \tilde{x}_p }_{X_0} \leq \tilde{r}_p$ if
\begin{equation} \label{eq:ellipsoid-tilde-2}
\left\{ 
\tilde{x}_p :
\tilde{x}_p ^ \top \tilde{x}_p \leq n_{\uptheta} \cdot \upphi
\right\}
\subseteq
\left\{ 
\tilde{x}_p :
\tilde{x}_p ^ \top X_0 \tilde{x}_p \leq \tilde{r}_p ^ 2
\right\},
\end{equation}
which leads to condition (b) of Proposition~\ref{prop:b-bar}.

According to Lemma~\ref{lemma:xs}, we have $e$ satisfies \eqref{eq:model-error}.
Therefore, $\norm{e}_{X_0} \leq r_{e}$ for all $d$ that $d \in \mathcal{D}$ if and only if $\frac{\df \, \norm{e}_{X_0} ^ 2}{\df \, t} < 0$, or equivalently
\begin{equation} \label{eq:B-dot-1}
\begin{bmatrix}
e \\
d \\
\hat{p}_{d}
\end{bmatrix}
^ \top 
\begin{bmatrix*}
A_z ^ \top X_0 + X_0 A_z & \star & \star \\
\hat{A}_{d} ^ {\top} X_0 & \0 & \star \\
\hat{B}_{d} ^ {\top} X_0 & \0 & \0
\end{bmatrix*}
\begin{bmatrix}
e \\
d \\
\hat{p}_{d}
\end{bmatrix}
< 0
,
\end{equation}
for all $e$, $d$, and $\hat{p}_d$ that
\begin{equation} \label{eq:s-condition-2}
\begin{aligned}
e ^ \top X_0 e \geq r_e ^ 2,
\quad
& 
d_{(i)} ^ {2}
\leq \bar{d}_{(i)} ^ 2,
&&
\text{for} \ i = 1, \, 2, \, \dots, \, n_d
\\
&
\hat{p}_{d(i)} ^ {2}
\leq \hat{q}_{d(i)} ^ {2},
\quad
&&
\text{for} \ i = 1, \, 2, \, \dots, \, n_{\uptheta}.
\end{aligned} 
\end{equation}
Using the S-procedure, \eqref{eq:B-dot-1} holds under the conditions in \eqref{eq:s-condition-2} if there exist $\hat{\upmu}_{e}, \, \hat{\upmu}_{d(1)}, \, \hat{\upmu}_{d(2)}, \, \dots, \, \hat{\upmu}_{d(n_d)} \geq 0$ and $\hat{\upmu}_{p(1)}, \, \hat{\upmu}_{p(2)}, \, \dots, \, \hat{\upmu}_{p(n_{\uptheta})} \geq 0$ such that for all $e$, $d$, and $\hat{p}_d$
\begin{equation} \label{eq:LMI-lyapunov-X-1}
\begin{bmatrix}
e \\
d \\
\hat{p}_d
\end{bmatrix}
^ \top 
\Omega_{X}
\begin{bmatrix}
e \\
d \\
\hat{p}_d
\end{bmatrix}
+ \sum_{i} ^ {n_d} \hat{\upmu}_{d(i)} \bar{d}_{(i)} ^ 2
- \hat{\upmu}_e
< 0,
\end{equation}
where 
\begin{equation} \label{eq:LMI-lyapunov-X-2}
\Omega_{X} \triangleq
\begingroup
\setlength\arraycolsep{2pt}
\begin{bmatrix}
A_z ^ \top X_0 + X_0 A_z + \frac{\hat{\upmu}_e}{r_e ^ 2} X_0 & \star & \star \\
\hat{A}_{d} ^ {\top} X_0 & - \hat{M}_{d} + \hat{C}_{d} ^ {\top} \hat{M}_p \hat{C}_{d} & \star \\
\hat{B}_{d} ^ {\top} X_0 & \0 & - \hat{M}_{p}
\end{bmatrix}
\endgroup
.
\end{equation}
Therefore, we can conclude that \eqref{eq:B-dot-1} holds if $\Omega_{X} \prec \0$ and $\hat{\upmu}_e = \sum_{i} ^ {n_d} \hat{\upmu}_{d(i)} \bar{d}_{(i)} ^ 2$, which are equivalent to condition (c) in Proposition~\ref{prop:b-bar}.

Consequently, $\mathbf{B}_\mathtt{C} (x_\mathtt{CL}) \leq \bar{\mathbf{B}}_\mathtt{C}$ if conditions (a), (b), and (c) in Proposition~\ref{prop:b-bar} are satisfied.
\end{proof}

Based on the conditions (a), (b), and (c) in Proposition~\ref{prop:b-bar}, we can use convex optimizations to minimize the values of $\bar{r}_\mathtt{CL}$, $\tilde{r}_p$, and $r_{e}$ such that we obtain a tight upper bound $\bar{\mathbf{B}}_\mathtt{C}$ for $\mathbf{B}_\mathtt{C} (x_\mathtt{CL})$. 
Since $\bar{x}_\mathtt{CL}$ is available to us, \eqref{eq:lmi-estimator-1} is LMI in $\upgamma$ and $\bar{r}_\mathtt{CL} ^ 2$. Therefore, we can use convex optimization
\begin{equation} \label{eq:optimization-V-bar-1}
\begin{aligned}
& \underset{\bar{\upphi}, \upgamma_0, \upgamma_1, \dots, \upgamma_{n_{\upgamma}}}{\mathtt{minimize} \quad}
& & && \bar{\upphi} \triangleq \bar{r}_\mathtt{CL} ^ 2 \\
& \mathtt{subject \ to} & & &&
\eqref{eq:lmi-estimator-1}, \qquad \sum_{i = 0} ^ {n_{\upgamma}} \upgamma_i = 1, \\
& & & && \upgamma_0, \ \upgamma_1, \ \dots, \ \upgamma_{n_{\upgamma}} \geq 0,
\end{aligned}
\end{equation}
for minimizing $\bar{r}_\mathtt{CL}$, which is also equivalent to $\mathbf{B}_\mathtt{C} (\bar{x}_\mathtt{CL})$ according to the definition of $\mathbf{B}_\mathtt{C} ( \star )$ in \eqref{eq:Bc}.
Since $\bar{x}_\mathtt{CL}$ is in terms of our identifier-based estimator variables $E_{h(j)}$ for $j = 1, \, 2, \, \dots, \, n_{h}$, we need to solve \eqref{eq:optimization-V-bar-1} in real-time for obtaining the optimal values of $\bar{r}_\mathtt{CL}$.
For high-dimensional systems, we may need to synthesize a lot of $Q_i$ matrices, which cause the computational complexity of \eqref{eq:optimization-V-bar-1} to increase.
Since $\bar{\upphi} = \underset{\upgamma \in \Gamma}{\min} \ \bar{x}_\mathtt{CL} ^ \top Q_{\upgamma} ^ {- 1} \bar{x}_\mathtt{CL}$ for $Q_{\upgamma} \triangleq \sum_{i = 0} ^ {n_{\upgamma}} \upgamma_i Q_i$ and $\sum_{i = 0} ^ {n_{\upgamma}} \upgamma_i = 1$, we can reduce the computational time by approximating the optimal solution of \eqref{eq:optimization-V-bar-1} through interpolation. 
Supposing that ${Q}_{n_{\gamma} + 1}, \, {Q}_{n_{\gamma} + 2}, \, \dots, \, {Q}_{n_{\gamma} + \bar{n}_{\gamma}}$ are interpolated from ${Q}_{0}, \, {Q}_{1}, \, \dots, \, {Q}_{n_{\gamma}}$ for $\bar{n}_{\gamma} > 0$, we define an approximated solution of \eqref{eq:optimization-V-bar-1} as
\begin{equation}
\bar{\upphi}_\mathtt{int} \triangleq \underset{i \in \mathrm{I}_\mathtt{int}}{\min} \ \bar{x}_\mathtt{CL} ^ \top Q_{i} ^ {- 1} \bar{x}_\mathtt{CL}
\end{equation}
where $\mathrm{I}_\mathtt{int} \triangleq \{ 0, \, 1, \, \dots, \, n_{\gamma} + \bar{n}_{\gamma} \}$. Then, we have $\bar{\upphi} \leq \bar{\upphi}_\mathtt{int}$ and $\bar{\upphi} = \lim\limits_{\bar{n}_{\gamma}\to\infty} \bar{\upphi}_\mathtt{int}$.
In our example in Sec.~\ref{sec:example-mimo}, we implement this approximation for a $4$-th order MIMO system with $n_{\gamma} = 6$ and $\bar{n}_{\gamma} = 84$, which reduces the computation time of \eqref{eq:optimization-V-bar-1} to $1.1 \ \mathrm{ms}$ in a $\mathrm{Python}$ simulation.

Similarly, we can use convex optimization
\begin{equation} \label{eq:optimization-V-bar-2}
\begin{aligned}
& \underset{\tilde{\upphi}}{\mathtt{minimize}}
& & && \tilde{\upphi} \triangleq \tilde{r}_p ^ 2 \\
& \mathtt{subject \ to} & & &&
\tilde{\upphi} \cdot \I - n_{\uptheta} \cdot \upphi \cdot X_0 \succeq \0,
\end{aligned}
\end{equation}
for minimizing $\tilde{r}_p$.
Unlike \eqref{eq:optimization-V-bar-1}, the optimization problem in \eqref{eq:optimization-V-bar-2} only has one variable $\tilde{\upphi}$, which can be solved easily through $1$-dimensional line search.

If we fix the values of $\hat{\upmu}_e$, \eqref{eq:lmi-estimator-3} becomes an LMI in $B_{z}$, $\hat{M}_d$, $\hat{M}_p$, and $\frac{1}{r_e ^ 2}$.
Through convex optimization
\begin{equation} \label{eq:optimization-estimator}
\begin{aligned}
& \underset{\upphi_e, B_{z}, \hat{M}_d, \hat{M}_p}{\mathtt{maximize}}
& &  \upphi_e \triangleq \frac{1}{r_e ^ 2} \\
& \mathtt{subject \ to} & & 
\eqref{eq:lmi-estimator-3}, \
\hat{\upmu}_e = \sum_{i} ^ {n_d} \hat{\upmu}_{d(i)} \bar{d}_{(i)} ^ 2, \\
&&&
\hat{\upmu}_{d(1)}, \ \hat{\upmu}_{d(2)}, \ \dots, \ \hat{\upmu}_{d(n_d)} \geq 0, \\
&&&
\hat{\upmu}_{p(1)}, \ \hat{\upmu}_{p(2)}, \ \dots, \ \hat{\upmu}_{p(n_{\uptheta})} \geq 0, \\
\end{aligned}
\end{equation}
we minimize the value of $r_e$.
This optimization problem is solved off-line and does not create runtime computation.
Since \eqref{eq:optimization-estimator} determines the value of $B_{z}$ for minimizing $r_{e}$, we can consider it a synthesis problem for our identifier-based estimator shown in \eqref{eq:estimator-h}.

\subsection{Switching Logic of $\Sigma_s$} \label{sec:switching}

Based on the value of $\bar{\mathbf{B}}_\mathtt{C}$, we can now define the switching logic of $\Sigma_s$ (Fig.~\ref{fig:switching}).
Let us define two thresholds $\ubar{\upvarepsilon}$ and $\bar{\upvarepsilon}$ (with $\upvarepsilon < \ubar{\upvarepsilon} < \bar{\upvarepsilon} \leq 1$).
$\Sigma_s$ switches from the original input $u = \hat{u}$ to $u = \Sigma_k (y)$ if $\bar{\mathbf{B}}_\mathtt{C} \geq \bar{\upvarepsilon}$ and switches back to $u = \hat{u}$ if $\bar{\mathbf{B}}_\mathtt{C} < \ubar{\upvarepsilon}$.
Note that $\bar{\mathbf{B}}_\mathtt{C}$ may not consistently decrease throughout the operation of $\Sigma_k$. 
However, according to Proposition~1, $x_\mathtt{CL}$ converges to residual set $\{ x_\mathtt{CL} : \mathbf{B}_\mathtt{C} (x_\mathtt{CL}) \leq \upvarepsilon \}$ when $u = \Sigma_k (y)$. 
Hence, under the control of $\Sigma_k$, the actual value of $\mathbf{B}_\mathtt{C} (x_\mathtt{CL})$ is ensured to be no greater than $\bar{\upvarepsilon}$ and eventually descends below $\ubar{\upvarepsilon}$ within a finite time.

As discussed in \cite{thomas2018safety}, the tuning of the thresholds $\bar{\upvarepsilon}$ and $\ubar{\upvarepsilon}$ has a significant impact on the behavior of the system. 
When $\ubar{\upvarepsilon}$ approaches $\bar{\upvarepsilon}$, the safety controller strictly enforces the constraint $\mathbf{B}_\mathtt{C} \leq \bar{\upvarepsilon}$, aiming for precise adherence to safety limits. 
On the other hand, when $\ubar{\upvarepsilon}$ approaches $\upvarepsilon$, the system returns to the residue set $\{ x_\mathtt{CL} : \mathbf{B}_\mathtt{C} (x_\mathtt{CL}) \leq \upvarepsilon \}$ after encountering the safety limits.

\begin{figure}[!tbp]
\centering
\includegraphics[width=0.50\textwidth]{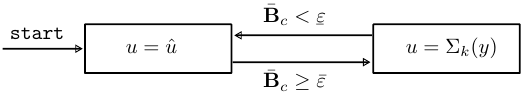}
\caption{$\Sigma_s$ switches from the original input $u = \hat{u}$ to $u = \Sigma_k (y)$ if $\bar{\mathbf{B}}_\mathtt{C} \geq \bar{\upvarepsilon}$ and switches back to $u = \hat{u}$ if $\bar{\mathbf{B}}_\mathtt{C} < \ubar{\upvarepsilon}$.}
\label{fig:switching}
\end{figure}

\begin{figure}[!tbp]
\centering
\includegraphics[width=0.20\textwidth]{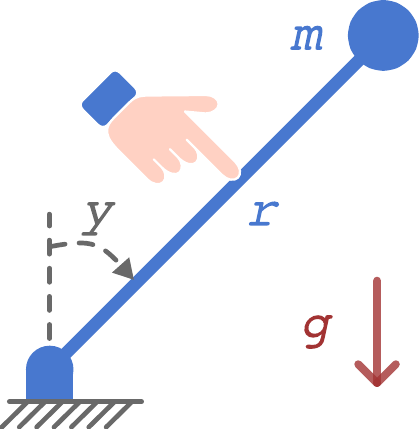}
\caption{In our first example, we consider an inverted pendulum system with $g = 9.8 \ \mathtt{m \cdot s ^ {-2}}$, $m = \frac{1}{9.8 ^ 2} \ \mathtt{kg}$, and $r = 9.8 \ \mathtt{m}$.}
\label{fig:example-siso}
\end{figure}

\begin{figure*}[!tbp]
\centering
\includegraphics[width=1.00\textwidth]{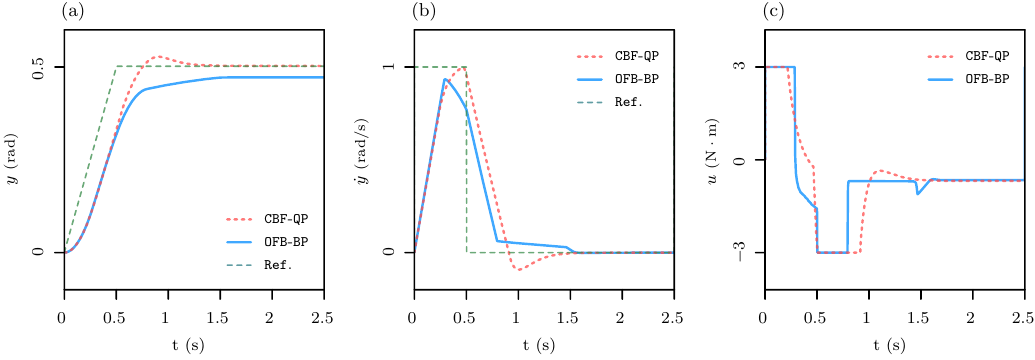}
\caption{{\bf Simulation Results of Inverted Pendulum Example}---(a)-(b) show the reference trajectories ($\mathtt{Ref.}$) of $y$ and $\dot{y}$ and the actual trajectories using the $\mathtt{CBF}$-$\mathtt{QP}$ method and our $\mathtt{OFB}$-$\mathtt{BP}$ method. (c) shows the actual input torque $u$ generated from the two methods, where we use the sigmoidal approximation in \eqref{eq:sigmoidal} with $\ubar{\upvarepsilon}=0.98$ and $\bar{\upvarepsilon}=1.00$ for the $\mathtt{OFB}$-$\mathtt{BP}$ method.}
\label{fig:result-siso}
\end{figure*}

Deviation of $\bar{\upvarepsilon}$ from $1$ reduces the actual size of the reachable state space under $\Sigma_s$, but it provides resilience against noise in the computation of $\bar{\mathbf{B}}_\mathtt{C}$. 
The gap between $\bar{\upvarepsilon}$ and $\ubar{\upvarepsilon}$ indirectly influences the rate of back-and-forth switching when the system approaches its limits.

In cases where the gap between $\bar{\upvarepsilon}$ and $\ubar{\upvarepsilon}$ is small, a sigmoidal approximation can be employed to reduce chattering. 
This approximation offers a more continuous control input $u$, mitigating the effects of rapid switching. 
Similar to the approach in \cite{hamayun2016integral}, a sigmoidal approximation can be defined as
\begin{equation} \label{eq:sigmoidal}
\begin{aligned}
u = \frac{1 - \upsigma}{2} \hat{u} + \frac{1 + \upsigma}{2} \Sigma_k (y) \quad \text{for} \quad
\upsigma = \frac{\updelta_\mathtt{C}}{\abs{\updelta_\mathtt{C}} + \frac{\updelta_{\upvarepsilon}}{20}}
\end{aligned}
\end{equation}
where $\updelta_\mathtt{C} \triangleq \bar{\mathbf{B}}_\mathtt{C} - \frac{\bar{\upvarepsilon} + \ubar{\upvarepsilon}}{2}$ represents the deviation of $\bar{\mathbf{B}}_\mathtt{C}$ from the average threshold, and $\updelta_{\upvarepsilon} \triangleq \frac{\bar{\upvarepsilon} - \ubar{\upvarepsilon}}{2}$ denotes half the difference between $\bar{\upvarepsilon}$ and $\ubar{\upvarepsilon}$.

\section{Case Studies and Benchmarking} \label{sec:example}

In this section, we illustrate the efficacy of our proposed output-feedback barrier pair ($\mathtt{OFB}$-$\mathtt{BP}$) method in safety control and barrier function estimation through examples of both SISO and MIMO systems. 
While our paper primarily addresses uncertain MIMO systems, SISO systems are also accommodated as special cases within our method. 
Specifically, the SISO example serves as a benchmark for comparison against the MIMO example, offering insight into the performance differences.
Furthermore, we compare our $\mathtt{OFB}$-$\mathtt{BP}$ approach with the CBF-based quadratic programming ($\mathtt{CBF}$-$\mathtt{QP}$) method introduced in \cite{xu2018constrained} within these examples.
Unlike our safety control scheme depicted in Fig.~\ref{fig:problem}, the $\mathtt{CBF}$-$\mathtt{QP}$ method in \cite{xu2018constrained} focuses on minimizing the discrepancy between the original input $\hat{u}$ of a system and an input $u$ designed to enforce control barrier function constraints \cite{nguyen2016exponential, ames2017control}.

\subsection{Inverted Pendulum System}

In the first example, we consider an inverted pendulum system (Fig.~\ref{fig:example-siso}, $g = 9.8 \ \mathtt{m \cdot s ^ {-2}}$, $m = \frac{1}{9.8 ^ 2} \ \mathtt{kg}$, $r = 9.8 \ \mathtt{m}$) with dynamics
\begin{equation}
\ddot{y} = \frac{g}{r} \cdot \sin (y) + \frac{1}{m \cdot r ^ 2} \cdot u + w
\end{equation}
where the variable $y$ represents the pendulum joint angle, $u$ denotes the torque input applied to the pendulum, and $w$ represents the disturbance caused by external forces exerted by a human. 
The regions $\mathcal{X}_s$, $\mathcal{U}$ and $\mathcal{D}$ are defined as
\begin{equation} \label{eq:example-siso-constraint}
\begin{aligned}
\mathcal{X}_s & \triangleq \big\{ \begin{bmatrix} \dot{y} & y \end{bmatrix} ^ \top : \ \ |y| \leq 0.5 \ \mathtt{rad}, \ \ |\dot{y}| \leq 1 \ \mathtt{rad \cdot s ^ {- 1}} \big\}, \\
\mathcal{U}   & \triangleq \big\{ u : |u| \leq 3   \ \mathtt{N \cdot m} \big\}, \\
\mathcal{D}   & \triangleq \big\{ w : |w| \leq 0.2 \ \mathtt{m \cdot s ^ {-2}} \big\}. 
\end{aligned}
\end{equation}
In the context where only the output $y$ and the input $u$ are available for measurement, the state model of this inverted pendulum system in the form of \eqref{eq:model} is
\begin{align}
\dot{x}_p
& = 
\overbracket{
\begin{bmatrix}
0 & 0 \\
1 & 0 
\end{bmatrix}
}^{\bar{A}}
x_p
+ 
\overbracket{
\begin{bmatrix}
\frac{1}{m \cdot r ^ 2} \\
0 
\end{bmatrix}
}^{\Theta_u}
u 
+
\overbracket{
\begin{bmatrix}
\uptheta_g \\
0
\end{bmatrix}
}^{\Theta_y}
y
+
\begin{bmatrix}
1 \\
0 
\end{bmatrix}
w 
\notag
\\
y & =
\underbracket{
\begin{bmatrix}
0 & 1 \\
\end{bmatrix}
}_{\bar{C}}
x_p
\label{eq:example-siso-model}
\end{align}
where $x_p \triangleq \begin{bmatrix} \dot{y} & y \end{bmatrix} ^ \top$ and $\uptheta_g = \frac{\sin (y)}{y} \cdot \frac{g}{r}$. For all $x_p \in \mathcal{X}_s$, we have $\frac{\sin (0.5)}{0.5} \cdot \frac{g}{r} \leq \uptheta_g \leq \frac{g}{r}$.

Through optimization \eqref{eq:optimization-volume}, we obtain an output-feedback barrier pair with a residue set $\{ x_\mathtt{CL} : \mathbf{B} (x_\mathtt{CL}) \leq \upvarepsilon \}$ of its barrier function with $\upvarepsilon = 0.8$. 
Then, we create an identifier-based estimator in the form of \eqref{eq:estimator-h}, which provides us a state estimate $\hat{x}_p$ to calculate a barrier function upper bound $\bar{\mathbf{B}}$ with $r_{e} = 0.04$.

The original system input $\hat{u}$ is employed as a reference tracking controller, starting from $y = 0$ and moving to $y = 0.5 \, \mathtt{rad}$ with $\dot{y} = 1 \, \mathtt{rad \cdot s^{-1}}$, after which it remains at $y = 0.5 \, \mathtt{rad}$. 
For the disturbance exerted by a human, we set $w = 0.2 \, \mathtt{m \cdot s ^ {-2}}$.

While the $\mathtt{CBF}$-$\mathtt{QP}$ method utilizes full state information, our $\mathtt{OFB}$-$\mathtt{BP}$ method relies solely on output information. 
The simulation results indicate that the $\mathtt{CBF}$-$\mathtt{QP}$ method fails to enforce the constraints of $y \leq 0.5 \, \mathtt{rad}$ and $\dot{y} \leq 1 \, \mathtt{rad \cdot s^{-1}}$ when the input saturation occurs, with $\abs{u} \leq 3 \, \mathtt{N \cdot m}$. 
Because of the LMI constraint in \eqref{eq:LMI-ub}, our $\mathtt{OFB}$-$\mathtt{BP}$ method ensures the safety of the pendulum trajectory at all times.

\subsection{Double Spring-Mass System} \label{sec:example-mimo}

\begin{figure}[!hb]
\centering
\includegraphics[width=0.40\textwidth]{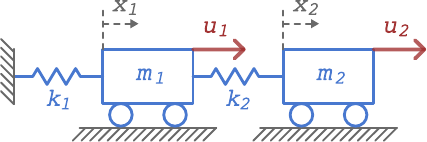}
\caption{In our second example, we consider a dual spring-mass system, where $m_1 = m_2 = 2$, $0.9 \leq k_1 \leq 1.0$ and $0.9 \leq k_2 \leq 1.0$.}
\label{fig:example-mimo}
\end{figure}

\begin{figure*}
\centering
\includegraphics[width=1.00\textwidth]{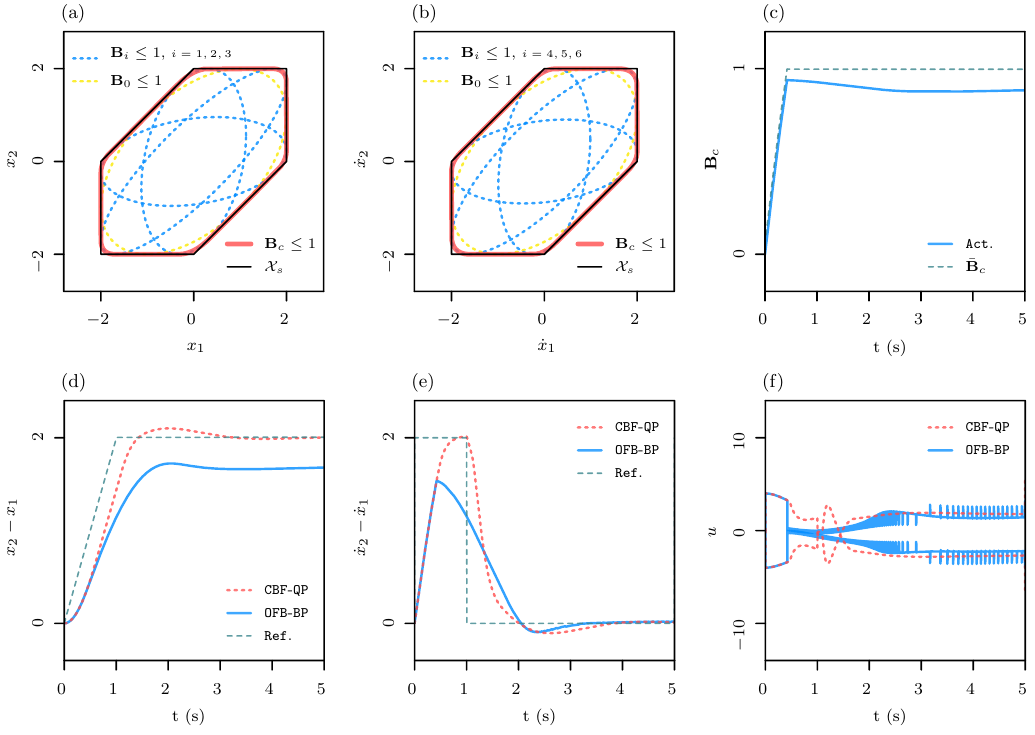}
\caption{{\bf Simulation Results of Double Spring-Mass Example}---(a)-(b) show the unit sub-level sets of $\mathbf{B}_0, \, \mathbf{B}_1, \, \mathbf{B}_2, \, \mathbf{B}_3, \, \mathbf{B}_4, \, \mathbf{B}_5, \, \mathbf{B}_6, \, \mathbf{B}_\mathtt{C}$ and $\mathcal{X}_s$ on the $x_1$-$x_2$ space and the $\dot{x}_1$-$\dot{x}_2$ space. (c)-(f) show the results of our safety control test with $v_1 = v_2 = -0.02$ and $k_1 = k_2 = 0.9$.
In particular, (c) shows the actual value ($\mathtt{Act.}$) and upper bound ($\bar{\mathbf{B}}_\mathtt{C}$) of $\mathbf{B}_\mathtt{C}$ in our $\mathtt{OFB}$-$\mathtt{BP}$ method, (d)-(e) show the reference trajectories ($\mathtt{Ref.}$) of $x_2 - x_1$ and $\dot{x}_2 - \dot{x}_1$ and the actual trajectories using the $\mathtt{CBF}$-$\mathtt{QP}$ method and our $\mathtt{OFB}$-$\mathtt{BP}$ method (f) shows the actual input torque $u$ generated from the two methods, where we use the sigmoidal approximation in \eqref{eq:sigmoidal} with $\ubar{\upvarepsilon}=0.98$ and $\bar{\upvarepsilon}=1.00$ for the $\mathtt{OFB}$-$\mathtt{BP}$ method.}
\vspace{-10pt}
\label{fig:result-mimo}
\end{figure*}

In the second example, we consider an uncertain double spring-mass system (Fig.~\ref{fig:example-mimo}, $m_1 = m_2 = 2$, $0.9 \leq k_1 \leq 1.0$, $0.9 \leq k_2 \leq 1.0$), where $x_1$ and $x_2$ are the positions of $m_1$ and $m_2$, the input $u \triangleq \begin{bmatrix} u_1 & u_2 \end{bmatrix} ^ \top$ denotes the adjustable forces exerted on $m_1$ and $m_2$. 
The output $y \triangleq \begin{bmatrix} y_1 & y_2 \end{bmatrix} ^ \top$ measures the spring forces of $k_1$ and $k_2$, while the external input $v \triangleq \begin{bmatrix} v_1 & v_2 \end{bmatrix} ^ \top$ represents the disturbances affecting the spring force measurements.
The regions $\mathcal{X}_s$, $\mathcal{U}$ and $\mathcal{D}$ are defined as
\begin{alignat}{4}
\mathcal{X}_s & \triangleq \big\{ \begin{bmatrix} \dot{x}_1 & x_1 & \dot{x}_2 & x_2 \end{bmatrix} ^ \top : 
&& |x_2 - x_1| \leq 2, \notag \\
&
&& |\dot{x}_2 - \dot{x}_1| \leq 2,  \label{eq:x-limit-example} \notag \\
&
&& |x_i| \leq 2, \ |\dot{x}_i| \leq 2, \ i = 1, \, 2 \big\}, \notag \\
\mathcal{U}   & \triangleq \big\{ \begin{bmatrix} u_1 & u_2 \end{bmatrix} ^ \top : |u_i| \leq 10, && \ i = 1, \, 2 \big\}, \\
\mathcal{D}   & \triangleq \big\{ \begin{bmatrix} v_1 & v_2 \end{bmatrix} ^ \top : |v_i| \leq 0.02, && \ i = 1, \, 2 \big\}. \notag
\end{alignat}
The state model of $\Sigma_p$ can be expressed in the form of \eqref{eq:model} as
\begin{align}
\dot{x}_p
& = 
\begingroup
\setlength\arraycolsep{3pt}
\overbracket{
\begin{bmatrix}
0 & 0 & 0 & 0 \\
1 & 0 & 0 & 0 \\
0 & 0 & 0 & 0 \\
0 & 0 & 1 & 0 
\end{bmatrix}
}^{\bar{A}}
x_p
+ 
\overbracket{
\begin{bmatrix}
\frac{1}{m_1} & 0 \\
0 & 0 \\
0 & \frac{1}{m_2} \\
0 & 0
\end{bmatrix}
}^{\Theta_{u}}
u 
\endgroup
+
\begingroup
\setlength\arraycolsep{1pt}
\overbracket{
\begin{bmatrix}
- \frac{1}{m_1} & \ \ \frac{1}{m_1} \\
0 & 0 \\
0 & - \frac{1}{m_2} \\
0 & 0
\end{bmatrix}
}^{\Theta_{y}}
(y - v) 
\endgroup
\notag
\\
y & = 
\begingroup 
\setlength\arraycolsep{3pt}
\underbracket{
\begin{bmatrix*}[c]
\ \ k_1 & 0 \\
  - k_2 & k_2
\end{bmatrix*}
}_{\Theta_{c}}
\underbracket{
\begin{bmatrix}
0 & 1 & 0 & 0 \\
0 & 0 & 0 & 1
\end{bmatrix}
}_{\bar{C}}
x_p
+ 
v
\endgroup
\label{eq:model-example-mimo}
\end{align}
where $x_p \triangleq \begin{bmatrix} \dot{x}_1 & x_1 & \dot{x}_2 & x_2 \end{bmatrix} ^ \top$.

Through optimizations \eqref{eq:optimization-volume} and \eqref{eq:optimization-direction-i}, we obtain $\Sigma_k$ with barrier functions $\mathbf{B}_0, \, \mathbf{B}_1, \, \cdots \, \mathbf{B}_6$ with $\upvarepsilon = 0.4$. 
These barrier functions allow us to build a composite barrier function $\mathbf{B}_\mathtt{C}$ such that $(\mathbf{B}_\mathtt{C}, \, \Sigma_k)$ becomes our final barrier pair. 
Fig.~\ref{fig:result-mimo}.(a)-(b) show the unit sub-level sets of $\mathbf{B}_0, \, \mathbf{B}_1, \, \cdots \, \mathbf{B}_6$ and $\mathbf{B}_\mathtt{C}$ projected onto two 2-D sub-spaces of $x_p$.
Through optimization \eqref{eq:optimization-estimator}, we create an identifier-based estimator in the form of \eqref{eq:estimator-h}, which provides us a state estimate $\hat{x}_p$ to calculate a barrier function upper bound $\bar{\mathbf{B}}_\mathtt{C}$ with $r_{e} = 0.11$.

The original system input $\hat{u}$ is employed as a reference tracking controller. 
Initially, the reference trajectory starts with $x_1 = x_2 = 0$. As $m_2$ transitions to $x_2 = 1$ with velocity $\dot{x}_2 = 1$, $m_1$ moves in the opposite direction mirroring the displacement of $m_2$. 
Consequently, the spring deflection $x_2 - x_1$ transitions from $0$ to $2$, while the spring deflection rate $\dot{x}_2 - \dot{x}_1$ is equal to $2$. 
Following this transition, $x_2 - x_1$ remains at $2$. 
In this whole process, we set $v_1 = v_2 = -0.02$ and $k_1 = k_2 = 0.9$.

In this example, both the $\mathtt{CBF}$-$\mathtt{QP}$ method and our $\mathtt{OFB}$-$\mathtt{BP}$ method utilize the estimated state $\hat{x}_p$ obtained from the identifier-based estimator, as shown in \eqref{eq:estimator-h}. 
Through simulation results, it is evident that the $\mathtt{CBF}$-$\mathtt{QP}$ method fails to satisfactorily enforce the constraints of $x_2 - x_1 \leq 2$ and $\dot{x}_2 - \dot{x}_1 \leq 2$, despite the absence of input saturation. 
This violation of safety constraints by the $\mathtt{CBF}$-$\mathtt{QP}$ method is attributed to the influence of the estimated state $\hat{x}_p$, which leads to constraint violations before the system reaches a steady state.

\subsection{Benchmarking Analysis} \label{sec:discussion}

When juxtaposing the outcomes of the SISO inverted pendulum example with those of the MIMO double spring-mass example, the latter reveals a slightly heightened level of conservatism, particularly noticeable when the system aims to reach its state-space safety limits under the operation of $u = \hat{u}$. 
This conservatism arises from various factors, including the extent of model uncertainty, reachability under input saturation, and the design of our barrier function. 
In contrast to the control barrier function methods in \cite{xu2015robustness, nguyen2016exponential, ames2017control, xu2018constrained}, our barrier function $\mathbf{B} (x_\mathtt{CL})$ does not yield the precise safety bounds defined in $\mathcal{X}_s$ for the invariant set in the $x_p$ space. 
It is possible for $x_p \in \mathcal{X}_s$ while $\mathbf{B} (x_\mathtt{CL}) > 1$. 
In such regions of the state space, there is no guarantee that the future trajectories of these plant states will remain within $\mathcal{X}_s$ when employing the safety controller $\Sigma_k$. 
To address this issue, the cost functions in optimizations \eqref{eq:optimization-volume} and \eqref{eq:optimization-direction-i} aim to maximize the invariant set of our quadratic barrier function and composite quadratic barrier function while minimizing the region where $x_p \in \mathcal{X}_s$ and $\mathbf{B}(x_\mathtt{CL}) > 1$. 
In our double spring-mass example, we extend the invariant set of the barrier function obtained from \eqref{eq:optimization-volume} in various directions to alleviate the over-conservatism of our barrier function. 
However, for systems with high dimensions, identifying the directions for expanding the composite barrier function may pose additional challenges.

The advantages of our $\mathtt{OFB}$-$\mathtt{BP}$ method compared to $\mathtt{CBF}$-$\mathtt{QP}$ are also evident. 
The results from the inverted pendulum example expose the limitations of the $\mathtt{CBF}$-$\mathtt{QP}$ method when confronted with tight input constraints, despite utilizing the true full state information. 
Specifically, the simulation results clearly demonstrate that the $\mathtt{CBF}$-$\mathtt{QP}$ method fails to satisfy the safety limit on $y$, which has a relative degree of $2$. 
This violation of the safety limit aligns with the observation made in \cite[Remark 11]{ames2017control}, which states that the CBF constraint with a relative degree equal to or greater than $2$ becomes invalid in the presence of input constraints. 
In \cite{rauscher2016constrained, barbosa2020provably, cortez2022safe}, various $\mathtt{CBF}$-$\mathtt{QP}$ methods are presented to tackle this challenge in MIMO Euler-Lagrangian systems. 
However, these methods depend on having access to the exact models and full state information of the systems. 
As shown in the dual spring-mass example, even without input saturation, the $\mathtt{CBF}$-$\mathtt{QP}$ method fails to enforce the specified constraints when using estimated state information. 
Recent advancements in CBF-based methods have introduced state estimators designed to estimate the values of CBFs, addressing challenges such as initial state errors \cite{wang2022observer} and stochastic disturbances \cite{clark2021control}. 
Nonetheless, our $\mathtt{OFB}$-$\mathtt{BP}$ method stands out for its ability to ensure safety even in scenarios involving both input saturation and partial state information.

\section{Conclusion}

In this paper, we addressed the safety control problem for an uncertain MIMO system with partial state information. 
Our approach involved the synthesis of a vector norm barrier function and a dynamic output feedback safety controller, working in tandem to ensure system safety. 
This dynamic output feedback safety controller guarantees the invariance of the barrier function even in the presence of uncertainties and disturbances inherent to the system dynamics. 
The development of our safety control synthesis method relied on a multi-objective output feedback scheme \cite{scherer1997multiobjective}, taking into account both uncertain dynamics and disturbances.

Furthermore, we introduced a robust barrier function estimator that leverages only input and output measurements to estimate an upper bound for the barrier function. 
This estimator enables reliable safety enforcement in scenarios where full state information is not available.

Importantly, our proposed methodology offers significant advantages over conventional CBF-based methods. 
While CBF-based approaches have shown effectiveness in enforcing safety constraints \cite{xu2015robustness, nguyen2016exponential, ames2017control, xu2018constrained}, they can be limited in scenarios involving tight input constraints or the use of estimated state information. 
In contrast, our method successfully addresses these limitations and ensures safety constraints are enforced, even under these adverse conditions. 
By incorporating the proposed barrier function estimator and dynamic output feedback safety controller, our approach achieves safety enforcement in uncertain MIMO systems with partial state information.

\balance

\bibliographystyle{IEEEtran}
\bibliography{main}

\end{document}